\documentclass[mnsc,nonblindrev]{informs3} 

\OneAndAHalfSpacedXI 

\def \hfillx {\hspace*{-\textwidth} \hfill}

\usepackage{hyperref}
\hypersetup{
  colorlinks=true,
  linkcolor=blue,
  citecolor=blue
}

\usepackage{bbm}
\usepackage{xfrac}
\usepackage[nice]{nicefrac}

\usepackage{natbib}
 \bibpunct[, ]{(}{)}{,}{a}{}{,}%
 \def\bibfont{\small}%
\TheoremsNumberedThrough     
\ECRepeatTheorems

\EquationsNumberedThrough    

\newcommand{\myparagraph}[1]{\bigskip\noindent\textbf{{#1}}}

\begin{document}


 \RUNAUTHOR{Ashlagi, Kerimov, Tamuz, and Zhao}

\RUNTITLE{The Power of Two in Token Systems}

\TITLE{The Power of Two in Token Systems}

\ARTICLEAUTHORS{%
\AUTHOR{Itai Ashlagi}
\AFF{Department of Management Science and Engineering, Stanford University, Stanford, California 94305, \EMAIL{iashlagi@stanford.edu}} 
\AUTHOR{S\"{u}leyman Kerimov}
\AFF{Jones Graduate School of Business, Rice University, Houston, Texas 77005, \EMAIL{kerimov@rice.edu}}
\AUTHOR{Omer Tamuz}
\AFF{Division of the Humanities and Social Sciences, Caltech, Pasadena, California 91125, \EMAIL{tamuz@caltech.edu}}
\AUTHOR{Geng Zhao}
\AFF{Department of Electrical Engineering and Computer Sciences, University of California, Berkeley, California 94720,  \EMAIL{gengzhao@berkeley.edu}}
} 


\ABSTRACT{%

In economies without monetary transfers, token systems serve as an alternative to sustain cooperation, alleviate free riding, and increase efficiency.  This paper studies whether a token-based economy can be effective in  marketplaces with  thin exogenous supply. We consider a marketplace in which at each time period one agent requests a service, one agent provides the service, and one token (artificial currency) is used to pay for service provision. The number of tokens each agent has represents the difference between the amount of service provisions and service requests by the agent. We are interested in the behavior of this economy when very few agents are available to provide the requested service. Since balancing the number of tokens across agents is key to sustain cooperation, the agent with the minimum amount of tokens is selected to provide service among the available agents. 
When exactly one random agent is available to provide service, we show that the token distribution is unstable. However, already when just two random agents are available to provide service, the token distribution is stable, in the sense that agents' token balance is unlikely to deviate much from their initial endowment, and agents return to their initial endowment in finite expected time. Our results mirror the power of two choices paradigm in load balancing problems. Supported by numerical simulations using kidney exchange data, our  findings suggest that token systems may generate efficient outcomes in kidney exchange marketplaces by sustaining cooperation between hospitals.}


\maketitle


%


\section{Introduction}

Token systems have been introduced as a market solution to economies in which  monetary transfers are undesirable or repugnant. Examples include trading favors in babysitting cooperatives \citep{sweeney1977monetary}, exchanging resources in peer-to-peer systems \citep{vishnumurthy2003karma}, and distributing food to food banks \citep{prendergast2016allocation}. The use of such artificial currency 
is intended to sustain cooperation, alleviate free riding, and increase efficiency. The incentive for agents to provide service (or resources) is to earn tokens and  the ability to spend them in future exchanges. This ability relies on the  liquidity of agents' availability to provide service upon request. This paper studies the behavior of token systems in thin marketplaces, where demand exceeds supply, and agents' availability for transactions is sparse.

One motivation for this study arises from kidney exchange. Many patients in need for a kidney transplant have a willing but incompatible living donor, which has led to the emergence of platforms that  arrange swaps in order to find a compatible donor \citep{roth2007efficient}.  Efficiency of these platforms relies on the  thickness of the pool, as a large fraction of these pools include pairs that are hard-to-match \citep{ashlagi2019matching}.  However,  free riding is a common behavior; in the US, hospitals often submit to national platforms the pairs they cannot match internally \citep{agarwal2019market},\footnote{\cite{agarwal2019market} documents that more than half of the exchanges in the US being arranged within hospitals.} and similar behavior was observed by countries in international collaborations.\footnote{See  \cite{ashlagi2021kidney} for discussion about collaborations and mergers.}  To alleviate free riding, token systems based on hospitals' contribution to the pool have been adopted in an ad-hoc manner,  most notably by  the National Kidney Registry in the US.\footnote{Numerous platforms count the number of altruistic donors (who have no intended donor) and end altruistic donor chains with patients in these hospitals. Some platforms attempt to equalize the number of donors and patients from each ``player" matched in each exchange \citep{biro2019modelling,mincu2020ip}.} A major concern  is whether such systems can reward back contributing hospitals successfully, given that  liquidity is low  due to the sparsity  of the kidney exchange pool. In this work, we ask if a token system can be effective in such sparse marketplaces and sustain cooperation between agents.

To address this question, we study a stylized model  in which agents   request and provide service over time, and tokens are used as artificial currency to pay for service provision. We consider an infinite horizon model with finitely many agents, where each agent has initially $0$ tokens, and agents are allowed to have negative number of tokens. At each time period, one randomly chosen agent requests a service, and a random subset of agents of size $d$ become  available to provide the service. One of these agents is selected to provide the requested service,  and the service requester pays one token to the service provider. The number of tokens each agent has represents the difference between the amount of service provisions and service requests by the agent. As balancing the number of tokens across agents is key to sustain cooperation, the agent with the fewest tokens  among the available agents is selected to provide the service. 

We are interested in the case in which $d$ is a small constant; in the context of kidney exchange, this aims to capture few match opportunities for a given patient-donor pair.  Our model is  based on \citet{johnson2014analyzing} and \citet{kash2015equilibrium}. Their models, however, assume that at each time period,   a constant fraction of agents are available to provide service, i.e., the service availability is not minimal. Intuitively, the larger the number of agents who are available to provide service, the easier it is to balance the amount of service provisions among agents.

Cooperation is important for token systems. For example, the Capitol Hill Babysitting Co-op, which aimed  to exchange babysitting hours between families  has crashed, since tokens' values depreciated \citep{sweeney1977monetary}. So in a healthy market,  agents should not accumulate too many tokens (which can lead to unraveling), or accumulate a large debt (which leads to free riding). 

Motivated by these potential frictions, we analyze the token distribution in our model, and identify conditions under which the token system is {\it stable}, in the sense that the Markov chain describing the process admits a stationary probability distribution (the formal definition is given  in \S \ref{sec:model}). Informally,  stability is described by two desired conditions. The first condition is a uniform boundedness condition; the  number of tokens each agent has does not deviate much from its initial state, with high probability. The second condition is a fairness condition; the number of tokens each agent has oscillates in such a way that agents return to their initial endowment in finite expected time.  In other words, the first condition ensures that agents will not accumulate or lose too many tokens. Thus, agents will not lose their incentive to cooperate. The second condition implies that the market clears in finite expected time, which alleviates  free riding and balances the number of requests and provisions for each agent continuously over time. In the context of the classical model of mean-variance preferences (e.g., see \cite{markowitz1952portfolio,tobin1958liquidity}), our stability conditions offer an attractive solution for strategic agents, where stable systems reduce agents' risk (for further discussion on strategic considerations, see \S \ref{sec:game}).

\myparagraph{Overview of results.} In the baseline case ($d=1$), only one agent is available to provide service. Thus at each time period, service requester and provider are chosen independently to exchange service for tokens. In this case the system is unstable, since the number of tokens each agent has behaves like a divergent or null recurrent random walk. 

When $d>1$, the service provider who has the minimal number of tokens is chosen among all available agents. In our main finding, we show that already when $d=2$, the system is stable. We show that the long-run probability of agents having more than $M$ tokens (or less than $-M$ tokens) is at most $O(1/M)$. We further show, using Lyapunov arguments, that this  probability is bounded by $O(a^M)$ for any $a>{2}/{3}$; here $O(\cdot)$ includes a constant that may depend on $n$. We further study stability in a large market, and show that as $n$ grows large, this probability is bounded by $(1/2)^M$. The large market analysis offers intuitive insights into the token dynamics.

 We also describe how the stylized model can be applied for kidney exchange in \S \ref{sec:kidneyapplication}, and we further perform numerical experiments to simulate the token distribution of participating hospitals using data from the National Kidney Registry (NKR) platform. The simulations reveal and validate that easy-to-match pairs have, on average, more than one, but very few compatible hard-to-match pairs. Despite this  sparsity, hospitals' tokens do not deviate much from the initial state, which is aligned with our predictions. It is worth noting that hospitals in our data vary significantly with respect to size and distribution of patient-donor pair characteristics. 

\myparagraph{Techniques.} This paper is  inspired and borrows from the  literature on the power of two choices \citep[see, e.g.,][]{mitzenmacher1996power, vvedenskaya1996queueing, azar1999balanced}. The main finding of this literature is  that  in load balancing problems, minimal choice can significantly reduce congestion both in dynamic and static settings. In our model, it is simple to analyze the system directly using a birth-death process when $n=2$. For $n \geq 3$, the system does not seem amenable to a complete analytical solution, but softer techniques allow us to show that it is stable. 

For finite $n$, we use the Lyapunov method to study a suitable exponential Lyapunov function's negative drift to establish tail bounds on the number of tokens agents have in the long-run. We further analyze the token distribution, and we explicitly characterize the tail bound as $n$ grows large using Kurtz's theorem on \textit{density dependent Markov chains} \citep[see][]{kurtz1981approximation}; the same techniques are used in \cite{mitzenmacher1996power} and \cite{vvedenskaya1996queueing} to analyze various load balancing problems including the supermarket model. While there are subtle differences between our model and the load balancing problems in the literature (as we explain below),  our work can be viewed as  another application of the power of two choices paradigm. Kurtz's theorem provides conditions under which a stochastic process can be approximated by a deterministic process in the limit. This allows us to show that the token distribution is balanced, and the number of tokens each agent has is unlikely to be far from its initial state.

\myparagraph{Related literature.} Numerous papers study  exchange economies using tokens or models for exchanging favors \citep{mobius2001trading,friedman2006efficiency,hauser2008trading,abdulkadiroglu2012optimal,kash2012optimizing,kash2015equilibrium,johnson2014analyzing}. This literature is concerned  with whether  cooperation can be sustained in equilibrium, and whether efficiency can be achieved. The main finding of this literature is that if agents are sufficiently patient, token  mechanisms  may lead to efficient outcomes.
 Closely related are  \cite{johnson2014analyzing}, \cite{kash2015equilibrium} and \citet{bo2018}. These papers study token systems as a strategic game in an infinite horizon with discounting, but with similar dynamics. The key difference is that these papers  assume that either all or a constant fraction of agents are available to provide service at each time period.
 \citet{kash2015equilibrium} study a model in which the service provider is chosen  independently from the token distribution, and show the existence of an equilibrium, in which  agents provide service when  their tokens is below some threshold. \citet{johnson2014analyzing} study the same model and show that under the  minimum token selection rule, agents  always provide service in equilibrium when punishments are feasible. 
  \citet{bo2018} extend their findings without using punishments. 
  Our paper contributes to this literature, by studying whether stability can be achieved with low liquidity (i.e., low availability of service).

Also related  is the literature on the power of two choices in load balancing problems \citep{azar1994line,azar1999balanced,mitzenmacher1996power,vvedenskaya1996queueing}. In this classic problem, $n$ balls are sequentially thrown into $n$ bins. The key finding is that if two bins are selected randomly and the ball is thrown to the bin with the lower load, then the fullest bin has exponentially fewer balls than if only one bin is chosen randomly in the throwing process. \citet{mitzenmacher1996power} and \cite{vvedenskaya1996queueing} find a similar result for the supermarket model, which is a dynamic queueing system, where the longest queue is much longer if customers choose randomly among all queues rather than choosing intelligently between two random queues. The state space of the process we are interested in can be obtained by truncating the state space of the supermarket model. We describe in detail the difference between our model and the supermarket model in Remark \ref{differencesupermarket} in \S \ref{bounded}. In short, using a queueing language, \citet{mitzenmacher1996power} focuses on the
	effect of the parameter $d$ on the expected time customers spend in the system and the length of the longest queue in the long-run. We are interested, however, in the overall distribution of all queue-lengths, and  provide a more detailed characterization for the system, which is needed for our analysis.
    In fact, \citet{azar1999balanced} also study an infinite horizon process in which at each time period, one random ball is removed from the bins, and one ball arrives which is assigned to one of the two random bins intelligently. Our stochastic process has subtle differences (using this language, instead of removing one random ball, we first pick a random bin, and then remove a ball), making  the machinery of \citet{azar1999balanced}  inapplicable.

\myparagraph{Notation.} We use $\mathbb{Z}_{\geq 0}$ and $\mathbb{Z}_{+}$ to denote the set of non-negative integers and the set of strictly positive integers, respectively.  We use $\mathbb{R}_{\geq 0}$ and $\mathbb{R}_{+}$ to denote the set of non-negative real numbers and the set of strictly positive real numbers, respectively. We write $\mathbb{E}_{\pi}[\cdot]$ and $\mathbb{P}_{\pi}(\cdot)$ to write the expectation and the probability with respect to a given distribution $\pi$, respectively.

\section{Model}
\label{sec:model}

There is a finite set of agents $\mathcal{A}=\{1, 2, \ldots , n\}$, $n \geq 2$. The time $t \in \mathbb{Z}_{\geq 0}$ is discrete. The number of tokens  agent $i \in \mathcal{A}$ has at time $t$ is denoted by $s_i^t \in \mathbb{Z}$. We assume that $s_i^0 = 0$ for all $i \in \mathcal{A}$.  Let $s^t \in \mathbb{Z}^n$  track the number of tokens agents have at time $t$.

Let $P=(p_i)_{{i \in \mathcal{A}}},Q=(q_i)_{{i \in \mathcal{A}}}$ be full-support probability measures over the set of agents. At each time period, nature picks one agent to become a {\it service requester} according to $P$. Let $d$ be a positive integer, which we call \textit{service availability density}. At each time period, nature picks \textit{available service providers} by selecting $d$ agents according to $Q$ independently and with replacement. Thus, at most $d$ agents are available to provide service at each time period. We say that  agents are \textit{symmetric} if $p_i = q_i = \frac{1}{n}$ for all $i \in \mathcal{A}$.

We refer the tuple $(n, P, Q, d)$ as the token system. We will analyze the behavior of the token system under a natural matching policy called the  \textit{minimum token selection rule} \citep[see, e.g.,][]{johnson2014analyzing}.
This policy, at each time period, selects the available provider with the lowest number of tokens as the {\it service provider} (ties are broken by choosing uniformly at random). At each time $t$, if  agent $i$ is the service requester and agent $j$ is the service provider,  $i$ pays one token to  $j$.   Note that an agent can provide service to herself.\footnote{This technical assumption is useful for the proofs. Moreover, it  is motivated by the kidney exchange setting, where patient-donor pairs from the same hospital can be matched internally in a centralized setting when hospitals merge their patient-donor pools. Our results hold with  minor modifications if an agent cannot serve  herself.} In this case, $s^{t+1}=s^t$. Otherwise, $s^{t+1}_i = s^t_i-1$, $s^{t+1}_j = s^t_j+1$, and $s^{t+1}_k = s^t_k$ for all $k \in \mathcal{A} \backslash  \{i,j\}$.  In either case, we have $\sum_{i \in \mathcal{A}} s_i^t = 0$ for all $t\geq0$. The case $d=1$ is the degenerate case, where the system simply selects one service requester and one service provider independently at random.

\myparagraph{Stability.} Under a token system $(n,P,Q,d)$, the state of amount of tokens $s^t$ evolves according to a Markov chain defined on the state space $\left\{s \in \mathbb{Z}^n\,:\, \sum_{i \in \mathcal{A}} s_i = 0\right\}$. Our assumptions that $P$ and $Q$ are full-support probability measures over $\mathcal{A}$ ensure that this Markov chain is irreducible. Furthermore, since there is a positive probability that the service requester is the service provider herself at each time period, the Markov chain is aperiodic. 

We say that a token system $(n,P,Q,d)$ is {\it stable} if this Markov chain has a stationary probability distribution. The reason we associate the existence of a stationary distribution with stability is the fact that it is equivalent to each of the following conditions:
\begin{itemize}
\item \textbf{(C1)} There is a uniformly small probability that the number of tokens owned or owed by any agent is large. Formally, there is a function $f \colon \mathbb{Z}_{\geq 0} \to [0,1]$ such that $\lim_{M \to \infty} f(M) = 0$, and for all times $t$ large enough and all agents $i \in \mathcal{A}$ it holds that $\mathbb{P}(|s^t_i| > M) < f(M)$. 
\item \textbf{(C2)} The expected time for the token system to clear is finite. Formally, let $T_0$ be the first time the system returns to $0$, i.e., $T_0 = \min\{t > 0\,:\, s^t=0\}$. Then $\mathbb{E}(T_0)$ is finite. Note that by the Markov property, this is also the expected time to return to $0$ after a later visit to $0$. 
\end{itemize}

While we do not explicitly incorporate strategic considerations in our model, we view stability as a necessary condition to sustain cooperation in an appropriately defined strategic game.\footnote{For example, in order to prove the existence of an equilibrium, where all agents play a threshold strategy, \cite{kash2015equilibrium} first determine the token distribution of agents in the long-run.} Since the chain is irreducible and aperiodic, stability implies that the stationary distribution is unique, and that over time the distribution of $s^t$ will converge to it. Conversely, if there is no stationary distribution, the distribution of $s^t$ will become more and more ``spread out'' as $t$ increases, with some agents either owning or owing a large number of tokens. We thus interpret stability as a necessary condition for the prevention of unraveling and free riding in a token system. In \S\ref{sec:game}, we offer a microeconomic justification for stability via mean-variance preferences, which is a simple model of risk aversion used in economics.


The key novel feature of our model is  that the service availability density $d$ is small, rather than being a constant fraction of $n$.  The larger the service availability density $d$ is, the easier it is to achieve stability, as weakly more agents can provide service at any time period, which provides more flexibility to balance service provisions among agents. Therefore unless stated otherwise, we focus on the case when the service availability is minimal, i.e.,  $d=2$.

\subsection{The case $d=1$} \label{subsection:d=1}

Note that when $d=1$, $s_i^t$ is a lazy random walk in one dimension for all $i \in \mathcal{A}$, and so by standard arguments the token system is not stable. Moreover, $s^t$ is a random walk on $\mathbb{Z}^{n-1}$ ($n-1$ dimensions), and hence,  the Markov chain will not be recurrent by P\'olya's Recurrence Theorem for all $n \geq 4$, so that the market will eventually stop clearing at all. 

 An immediate corollary is that the token system is also not stable under the \textit{random tie-breaking selection rule}, where the rule selects the service provider, even if all providers are available  ($d=n$), uniformly at random (see Corollary \ref{randomtiebreaking} and all missing details in Appendix \ref{appendix:section2}). 

\section{Two agents}
\label{sec:2player}
In this section, we analyze the token distribution when there are only 2 agents ($n=2$). Although understanding the token distribution for this case is simple, the analysis will provide insights regarding the token distribution for the general case. Informally, the best concentration around the initial point is achieved when $n=2$; as the number of agents increases, the ``distance" of the token distribution from the initial point increases as well.

\begin{proposition}\label{twoagentsthm}
For $d \geq 2$, the token system with $2$ agents is stable if and only if $q_1^d < p_1$, $q_2^d < p_2$. In this case, let $\pi$ be the the steady-state distribution of the Markov chain $(s^t:t\geq0)$. Then for all $M \in \mathbb{Z}_{+}$, we have
\begin{align*}
  \mathbb{P}_{\pi}(|s_i^t| > M)=  \dfrac{\Bigg( \dfrac{p_2
      q_1}{p_1 - q_1^d} \Bigg(\dfrac{p_2 q_1^d}{p_1(1-q_1^d)} \Bigg)^M
    + \dfrac{p_1 q_2}{p_2 - q_2^d} \Bigg(\dfrac{p_1
      q_2^d}{p_2(1-q_2^d)} \Bigg)^M \Bigg)}{\Bigg( 1 + \dfrac{p_2
      q_1}{p_1 -q_1^d}  + \dfrac{p_1 q_2}{p_2-q_2^d}  \Bigg)},
\end{align*}
for $i=1,2$. Moreover, the expected time between two successive occurrences of the initial state $(0,0)$ is given by $  1 + \dfrac{p_2 q_1}{p_1 -q_1^d}  + \dfrac{p_1 q_2}{p_2-q_2^d}.$ 

\end{proposition}
The proof is straightforward and given in Appendix \ref{appendix:section3}.\footnote{We note that similar results hold even with less amount of service availability in the following sense. Suppose that at each time period, at most $2$ agents are available to provide service with probability $0<\beta<1$, and only one agent is available to provide service with probability $1-\beta$, independently. The analysis of this system is very similar and given in Appendix \ref{appendix:section3}. This suggests that even with having a weaker power of tie-breaking, two agents can still trade favors while keeping the token system stable.} One implication of Proposition \ref{twoagentsthm} is that the probability of owning or owing a large number of tokens decays exponentially, i.e.,   $f(M) = O(a^M)$, where the constant $a \in (0,1)$ can be found following Remark \ref{constantaremark} in Appendix \ref{appendix:section3}.

Proposition \ref{twoagentsthm} identifies the level of asymmetry that can be tolerated  between service request and service provision (within and across agents) rates. Moving forward, we focus on the symmetric case.  Note that for symmetric agents and $d\geq 2$, Proposition~\ref{twoagentsthm} implies that the system is stable.  When furthermore $d=2$, we get that
\begin{align*}
  \mathbb{P}_{\pi}(|s_i^t| > M)=  \frac{2}{3}\left(\frac{1}{3}\right)^M,
\end{align*}
and that $\mathbb{E}(T_0)=3$. In the general case, we will argue that $a=1/3$ is the best rate one can hope for.

\section{The general case}
\label{sec:general}

We analyze here the symmetric case with any number of agents.  The results for the general case can be summarized as follows:
\begin{enumerate}

\item In Section \ref{recurrence}, we show that stability holds for all $n \geq 2$ and $d \geq 2$. We further show that (C1) holds with $f(M)= 5/M$ for all $n \geq 2$ (Theorem \ref{thm2}). In particular, the probability that a given agent has a large number of tokens decays to zero in a rate that does not depend on the number of agents. 

\item In Section \ref{exponential_tail}, we show that the tail bound in (C1) can be strengthened with $f(M) = O((\frac{2}{3}+\delta)^M)$ for any $\delta > 0$, where $O(\cdot)$ hides an implicit constant that may depend on $n$ (Theorem \ref{expbound}). The proof relies on  Lyapunov arguments and is given in Appendix \ref{proofexpbound}. 

\item In Section \ref{bounded}, we show that as $n$ grows large,  (C1) holds with $f(M)=(1/2)^M$ (Theorem \ref{infiniteconstant}). The proof organization for Theorem \ref{infiniteconstant} is as follows. We first present \textit{density dependent Markov chains} and Kurtz's theorem. We then model our system as a density dependent Markov chain, and use Kurtz's theorem to characterize our system in the limit via a system of ordinary differential equations. Finally, we study the solution of this system to prove Theorem \ref{infiniteconstant}.

 \item We  conjecture that  (C1) holds for any $n \geq 2$ with $f(M)=a^M$, where $a \in [1/3,1/2]$ (Conjecture \ref{conjecturemonotinicity}).
\end{enumerate}

\subsection{Stability}\label{recurrence}

In general, it seems difficult to determine whether a given system $(n,P,Q,d)$ is stable. Indeed, the results of the previous section show that already for $n=2$, this can be highly sensitive to the precise values of $P$ and $Q$. The next result shows that in the symmetric case, stability is achieved for any $n \geq 2$, assuming $d \geq 2$.
\begin{theorem}\label{thm2}
The token system is stable for any $d\geq2$ when the agents are symmetric. Furthermore, (C1) holds with $f(M) = 5/M$.
\end{theorem}

The proof of Theorem \ref{thm2} is deferred to Appendix \ref{appendixthm1proof}. By analyzing the expected one-step transition difference of the potential function $\sum_{i=1}^n \mathbb{E}[(s_i^t)^2]$, we show that the Markov chain $(s^t:t\geq 0)$ has a stationary distribution, and $\mathbb{E}[|s_i^t|] \leq 5$ for all $i \in \mathcal{A}$ and for all $t$ large enough by exploiting a recursion on the expected one-step transition difference of the potential function. Theorem \ref{thm2} then follows from Markov's inequality.

\endproof

\subsection{Exponential decay}\label{exponential_tail}

We next offer a refined result, showing that under symmetry conditions, the tail bound on token profiles decays exponentially. This improves upon the polynomial bound previously established in Theorem~\ref{thm2}.

\begin{theorem}\label{expbound}
    For the token system with $d\ge 2$ and symmetric agents, (C1) holds with $f(M) = O(a^M)$ for any $a > 2/3$, where $O(\cdot)$ hides implicit constant factors possibly depending on $n$.
\end{theorem}

Unlike in Theorem~\ref{thm2}, the constant in the bound is implicit and may depend on the size $n$ of the system. Since we are unable to explicitly characterize the implicit constant, it is not clear under which regimes the exponential bound dominates the polynomial bound in Theorem~\ref{thm2} and vice versa. The proof of Theorem \ref{expbound} requires the following key observation.  For notational convenience, we write 
    $Z_\lambda^t := \sum_{i=1}^n e^{\lambda s_i^t}$
and 
    $E_\lambda^t := \mathbb{E}[Z_\lambda^t]$.

\begin{lemma}\label{lemma1:lyapunov}
    For $\lambda > 0$, if both $\limsup_{t\to\infty} E_\lambda^t < \infty$ and $\limsup_{t\to\infty} E_{-\lambda}^t < \infty$, then (C1) holds with $f(M)=O(e^{-\lambda M})$. 
\end{lemma}
\proof{Proof of Lemma~\ref{lemma1:lyapunov}.}
Note that whenever $|s_j^t| > M$ for some $j \in \mathcal{A}$, we must have $Z_\lambda^t + Z_{-\lambda}^t = \sum_{i=1}^n e^{\lambda|s_i^t|} + e^{-\lambda|s_i^t|} > e^{\lambda M}$. By Markov's inequality, 
\begin{equation*}
    \mathbb P(|s_i^t| > M) \le \mathbb P\Big(\sum_{i=1}^n e^{\lambda|s_i^t|} + e^{-\lambda|s_i^t|} > e^{\lambda M}\Big) \le e^{-\lambda M} (E_\lambda^t + E_{-\lambda}^t).
\end{equation*}
The quantity $E_\lambda^t + E_{-\lambda}^t$ is uniformly bounded for all $t$ sufficiently large, finishing our proof.
\Halmos\endproof

Lemma~\ref{lemma1:lyapunov} prepares the ground for us to apply the Lyapunov method with $E_\lambda^t$ as a family of potential functions; it will suffice to establish the asymptotic finiteness of $E_\lambda^t$ for $\lambda$ in the interval $(-\log 1.5, \log 1.5)$, where $\log$ is the natural logarithm. Our proof will proceed in two steps. First, we show that $\limsup_{t\to\infty} E_\lambda^t < \infty$ for $\lambda$ sufficiently close to zero (possibly depending on $n$). Then we remove this dependency and widen the range of $\lambda$ to cover the desired interval using a fixed point argument, exploiting the log-convexity of the mapping $\lambda \mapsto E_\lambda^t$. The proof of Theorem \ref{expbound} is deferred to Appendix \ref{proofexpbound}.

\subsection{Stability in the large}\label{bounded}

In this section, we study the large market setting to offer intuition regarding the dynamics of the market by focusing on the densities of agents with a given number of tokens. What we establish next suggests that stability is  ``well-behaved" in the large; as $n$ grows large, (C1) holds with $f(M)=(\frac{1}{2})^M$.

Let $\pi$ be the steady-state distribution of the Markov chain $(s^t:t \geq 0)$ granted by Theorem~\ref{thm2}. For any $M \in \mathbb{Z}_+$, define $p_{n,M}:=\mathbb{P}_{\pi}(|s_1^t| \leq M)$ when there are $n$ symmetric agents. 

\begin{theorem}\label{infiniteconstant}
 $\lim_{n \rightarrow \infty} p_{n,M} \geq 1 - (1/2)^M$ for all $M \in \mathbb{Z}_{+}$.
\end{theorem}

We will use Kurtz's theorem on \textit{density dependent Markov chains}, which allows us to analyze the stochastic process as a deterministic process in the limit. This analysis helps us to characterize the steady-state distribution of the system as the number of agents grows large and thus, proving Theorem \ref{infiniteconstant}.

\myparagraph{Density dependent Markov chains.} 
We begin with the definition of density dependent Markov chains, which is given in \citet{kurtz1981approximation} for finite dimensional systems and extended by \cite{mitzenmacher1996power} to  countably infinite dimensional systems. Let $\mathbb{Z}^{*}$ be either $\mathbb{Z}^m$ for some finite dimension $m$, or $\mathbb{Z}^\mathbb{N}$, and similarly define $\mathbb{R}^{*}$. Let $L \subseteq \mathbb{Z}^*$ be the set of possible non-zero transitions of the system. For each $\vec{l} \in L$, define a nonnegative function $\beta_{\vec{l}} \colon \mathbb{R}^{*} \rightarrow [0,1]$.
\begin{definition}\label{definition:kurtz}
A sequence (indexed by $n$) of continuous time Markov chains $(X_n(t): t \geq 0)$ on the state spaces $S_n = \left\{ \vec{k} / n : \vec{k} \in \mathbb{Z}^{*} \right\}$ is a density dependent Markov chain if there exists a $\beta_{\vec{l}} \colon \mathbb{R}^{*} \rightarrow [0,1]$ such that for all $n$ the transition rate of $X_n$ is given by
$q_{x,y}^{(n)} = n \beta_{n(y-x)} (x)$, $x,y \in S_n.$
\end{definition}

In Definition \ref{definition:kurtz}, the index $n$ can be interpreted as the total population or volume of the system, and the components of $\vec{k}/n$ can be interpreted as the densities of different types present in the system. The $\beta_{\vec{l}}\;(x)$ can be interpreted as the probability of transition $\vec{l}$ from $x \in S_n$ to $y \in S_n$, where $nx + \vec{l} = ny$. Given a density dependent Markov chain $X_n$ with transition rates $q_{{x,y}}^{(n)} = q_{{\vec{k},\vec{k}+\vec{l}}}^{(n)} = n \beta_{\vec{l}} \;(\vec{k}/n)$, define  $F(x) = \sum_{\vec{l} \in L} \vec{l} \beta_{\vec{l}} \;(x)$. The following theorem is key in our analysis: 

\begin{theorem}[Kurtz's theorem \citep{kurtz1981approximation,mitzenmacher1996power}]\label{kurtztheoremm}\label{kurtzzthm}
Suppose we have a density dependent Markov chain $X_n$ (of possibly countably infinite dimension) satisfying the Lipschitz condition $|F(x) - F(y)| \leq M |x-y|$
for some constant $M$. Further suppose $\lim_{n \rightarrow \infty} X_n(0) = x_0$, and let $X$ be the deterministic process:
\begin{equation}\label{eq:deterministic}
X(t) = x_0 + \int_0^t F(X(u))du, \; t \geq 0.
\end{equation}
Consider the path $\left\{X(u) : u \leq T \right\}$ for some fixed $T \geq 0$, and assume that there exists a neighborhood $K$ around this path satisfying
\begin{equation}\label{eq:jump}
\sum_{\vec{l} \in L} |\vec{l}| \sup_{x \in K} \beta_{\vec{l}}\; (x) < \infty.
\end{equation}
Then
$ \lim_{n \rightarrow \infty} \sup_{u \leq T} |X_n(u) - X(u)| = 0 \; \text{almost surely.}$
\end{theorem}
The Lipschitz condition ensures the uniqueness of the solution for the differential equation $\dot{X} = F(X)$, which follows by taking the derivative of (\ref{eq:deterministic}) with respect to $t$.\footnote{\label{footnote1}For example, see Lemma 4.1.6 in \cite{abraham2012manifolds}.} Condition (\ref{eq:jump}) ensures that the jump rate is bounded in the process. Kurtz's theorem implies that as $n \rightarrow \infty$, the behavior of a density dependent Markov chain can be characterized by the deterministic process given in (\ref{eq:deterministic}), where the convergence holds on a finite time interval $[0,T]$ for an arbitrary $T$. We next  model and study   our  system as a density dependent Markov chain, which we  refer to as the \textit{finite model}. 

\myparagraph{The finite and infinite models.} Let us model the system with $n$ symmetric agents as a density dependent Markov chain and denote it by $(X_n(t), t \geq 0)$. Note that in Definition \ref{definition:kurtz}, the $\beta_{\vec{l}}$'s are independent of $n$, and the transition rates are linear in $n$. In order to fit our system to this definition, we  assume that each agent has an exponential clock with rate $1$. The ticking of agent $i$'s clock corresponds to a service request by $i$, and the service provider is selected immediately using the minimum token selection rule. Note that because of the memoryless property and the continuity of the distribution that governs the clocks, agents request service uniformly, and exactly one agent requests service at a time. As a slight abuse of notation, let $n_i (t)$ be the number of agents with $i$ tokens at time $t$, $m_i (t)$ be the number of agents with at least $i$ tokens at time $t$, and $z_i (t) := m_i(t) / n$ be the fraction of agents with at least $i$ tokens at time $t$. Let $\vec{z}(t) = (..., z_{-2}(t), z_{-1}(t), z_0 (t), z_1(t), z_2(t), ...)$, and we drop the time index $t$ when the meaning is clear. We represent the state of $X_n$ by $\vec{z} = \vec{k}/n \in \mathbb{Z}^\mathbb{N}/n$. Let us call this process the \textit{finite model}. Note that the initial state of $X_n$ is $\vec{z}(0) = (..., 1, 1, 1, 0, 0, ...)$, where $z_i(0) = 1$ for all $i \leq 0$, and $z_i(0) = 0$ for all $i \geq 1$.

Next, we describe the transition probabilities $\beta_{\vec{l}}$'s. The set of possible non-zero transitions from $\vec{k}=n\vec{z}$ is $L = \left\{ e_{ij} : i,j \in \mathbb{Z}, i \neq j \right\}$, where $e_{ij}$ is an infinite dimensional vector of all zeros except the $i$'th index (which corresponds to the index of $z_i$) is $-1$ and the $j$'th index (which corresponds to the index of $z_j$) is $1$. Note that after transition $e_{ij}$ occurs, $n z_i$ decreases by $1$ and $nz_j$ increases by $1$ simultaneously. Hence, the transition $e_{ij}$ corresponds to the event when an agent with $i$ many tokens requests service and an agent with $j-1$ many tokens provides service. Since the probability that the service requester has $i$ many tokens is $z_{i} - z_{i+1}$, and the probability that the service provider has $j$ many tokens is $z_{j}^d - z_{j+1}^d$, we have $\beta_{e_{ij}}(\vec{z}) = (z_i - z_{i+1})(z_{j-1}^d - z_j^d)$.\footnote{The probability that all available agents have at least $j$ many tokens is $z_j^d$, and we subtract the probability that all available agents have at least $j+1$ many tokens, which is $z_{j+1}^d$.}
Denote the \textit{infinite model} by $X$, which is the limit of the finite model $X_n$, i.e., $X = \lim_{n \rightarrow \infty} X_n$. Since $X$ is characterized by the deterministic process (\ref{eq:deterministic}), we need to analyze the components of $F(x)$. Note that the $i$'th component of $F(x) = \sum_{\vec{l} \in L} \vec{l} \beta_{\vec{l}}\; (x)$ (which corresponds to $z_i$) is $\sum_{j \in \mathcal{A} \backslash \{i\}} (z_j - z_{j+1})(z_{i-1}^d - z_i^d) - \sum_{j \in \mathcal{A} \backslash \{i\}} (z_i - z_{i+1})(z_{j-1}^d - z_j^d)$, which simplifies to
\begin{equation}
(1 - z_i + z_{i+1})(z_{i-1}^d - z_i^d) - (1-z_{i-1}^d+z_i^d)(z_{i} - z_{i+1}) = (z_{i-1}^d - z_i^d) - (z_i - z_{i+1}).\label{eq:diffeq}
\end{equation}

\begin{remark}\label{differencesupermarket}
In the well-known supermarket model, customers arrive according to a Poisson process with rate $\lambda n$, $\lambda < 1$, where there are $n$ servers that serve according to the FIFO (first in, first out) rule. An arriving customer considers only a constant number ($d$) of servers independently and uniformly at random from the $n$ servers with replacement, and she joins the queue that contains fewest customers (any ties are broken arbitrarily). The service time for each customer is distributed exponentially with rate 1. Using this language, our model corresponds to this dynamic system, where the total number of customers is fixed throughout the process, which restricts the state space of the underlying process drastically.  Our system belongs to the family of dynamic systems referred as the closed model in \cite{mitzenmacher1996power}. Our technical contributions here are twofold. First, the main focus of this literature is on the expected longest-queue length and its dependence on the service availability parameter ($d$), whereas we refine this goal and further study tail-bounds on queue-lengths. Second, as stated in \cite{mitzenmacher1996power}, the steady-state probability that a certain queue has more than $M$ customers can be found by solving certain constraints that are derived from Theorem \ref{kurtztheoremm}, and in general, solution to these constraints does not appear to have a closed form. As we argue next, we will be able to bound the steady-state probabilities without obtaining a closed form solution.   
\end{remark}

\myparagraph{Towards the proof of Theorem \ref{infiniteconstant}.}
Now that we have represented our system using the finite model, we are ready to prove Theorem \ref{infiniteconstant}. The proof is organized as follows. We first show that the conditions of Theorem \ref{kurtzzthm} hold. Then using Theorem \ref{kurtzzthm}, we obtain the system of ordinary differential equations that characterize the infinite model. This characterization lets us  represent the probability of interest in (C1) as $n$ grows large using $\pi_0$ (the fraction of agents that have at least 0 tokens in the long-run). Finally, we find lower and upper bounds for $\pi_0$ to conclude.

Condition (\ref{eq:jump}) is clearly satisfied since the magnitude of any jump is bounded, and the jump rate is bounded above by $1$ for any state. We also show in Appendix \ref{appendix:section4} that the Lipschitz condition of Theorem \ref{kurtztheoremm} holds with $M=2+2d$. By differentiating (\ref{eq:deterministic}) with respect to $t$ and using (\ref{eq:diffeq}), we get the following system of ordinary differential equations that characterizes the infinite model:
\begin{equation}\label{eq:diffeq1}
\frac{dz_i}{dt} = (z_{i-1}^d - z_i^d) - (z_i - z_{i+1}) \;\; \text{for all} \;\;  i \in \mathbb{Z}.
\end{equation}
Intuitively, (\ref{eq:diffeq1}) can be interpreted as follows. Let us consider the expected change in $m_i$ (the number of agents with at least $i$ many tokens) over a small time interval $dt$. First note that a transition occurs whenever one of the agent's exponential clock ticks, which happens with rate $ndt$. Under such transition, $m_i$ increases by $1$ if an agent with $i-1$ many tokens is selected as the service provider, which happens with probability $z_{i-1}^d - z_i^d$. $m_i$ decreases by $1$ if an agent with $i$ many tokens is selected as the service requester, which happens with probability $z_i - z_{i+1}$. Hence, the expected increase in $m_i$ is $(z_{i-1}^d - z_i^d)ndt$, and the expected decrease in $m_i$ is $(z_i - z_{i+1})ndt$, which gives $dm_i = (z_{i-1}^d - z_i^d)ndt - (z_i - z_{i+1})ndt$, and since $m_i / n = z_i$, dividing both sides  by $ndt$ gives (\ref{eq:diffeq1}).

Define an \textit{equilibrium point}, which is a point  $\vec{a}$ such that if $\vec{z}(t') = \vec{a}$, then $\vec{z}(t) = \vec{a}$ for all $t \geq t'$.
Denote the equilibrium point of the infinite model by $\vec{\pi}$, and assume $d=2$ for simplicity from now on (the following arguments can be easily generalized for $d>2$). Clearly $\vec{\pi}$ is an equilibrium point of the infinite model if and only if $\frac{d\pi_i}{dt} = 0$ for all $i \in \mathbb{Z}$. Moreover,  since agents start with 0 tokens and exchange one token at each transition, the expected number of tokens agents have is $0$, and it can be written as follows: 
\begin{equation}
\sum_{i \in {Z}} i \cdot \frac{n_i}{n} = \sum_{i \geq 1} i \cdot \frac{ n_i}{n} + \sum_{i \leq 0} i \cdot \frac{n_i}{n} = \sum_{i \geq 1} \frac{m_i}{n} - \sum_{i \leq 0} \frac{n - m_i}{n} = \sum_{i \geq 1} z_i -  \sum_{i \leq 0}(1 - z_i) = 0. \label{eq:expnumofscrips}
\end{equation}
Using \eqref{eq:diffeq1} and \eqref{eq:expnumofscrips},  $\vec{\pi}$ can be found by solving the following system of equations:
\begin{equation}\label{eq:det1}
 (\pi_{i-1}^2 - \pi_i^2) - (\pi_i - \pi_{i+1}) = 0 \;\; \text{for all} \;\; i \in \mathbb{Z},
\end{equation}
\begin{equation}\label{eq:det2}
\sum_{i \geq 1}\pi_i -  \sum_{i \leq 0} (1 - \pi_i) = 0.
\end{equation}

\noindent Note that (\ref{eq:det1}) implies $\pi_{i+1} - \pi_i^2 = \pi_0 - \pi_{-1}^2$ for all $i \in \mathbb{Z}$. Since $\lim_{i \rightarrow \infty} \pi_i = 0$, we have $\pi_0 = \pi_{-1}^2$, and inductively we have the following relation:
\begin{equation}\label{eq:rec2}
 \pi_{i+1} = \pi_i^2 \;\; \text{for all} \;\; i \in \mathbb{Z}.
\end{equation}

\noindent Using (\ref{eq:rec2}), (\ref{eq:det2}) becomes $\sum_{i \geq 1} \pi_0^{2^i} - \sum_{i \geq 0} (1 - \pi_0^{2^{-i}}) = 0$. Such series are known as \textit{lacunary series}, where the function has no analytic continuation across its disc of convergence (see Hadamard's Gap Theorem). There is no closed form expression for such series to the best of our knowledge and thus, we are unable to find the equilibrium point explicitly. Note that in the long-run, the probability that $|s_i^t| \leq M$ for any $i \in \mathcal{A}$ is equal to $\pi_{- M} - \pi_{M + 1}.$ Using (\ref{eq:rec2}), proving that for all $M \in \mathbb{Z}_{+}$, $\lim_{n \rightarrow \infty} p_{n,M} \geq 1 - a^M$ for some $a \in (0,1)$, is equivalent to proving that the following inequalities hold:
\begin{equation}\label{closedformeq}
\pi_0^{2^{-M}} - \pi_0^{2^{M+1}} \geq 1 - a^{M} \;\; \text{for all} \;\; M \in \mathbb{Z}_{+}.
\end{equation}

\begin{lemma}\label{loww}
We have $\frac{1}{2} < \pi_0 < \frac{3}{4}$.
\end{lemma}

\noindent The proof of Lemma \ref{loww} is given in Appendix \ref{appendix:section4}, and now we use it to prove Theorem \ref{infiniteconstant}.

\proof{Proof of Theorem \ref{infiniteconstant}.} We will show that $g(M) := \pi_0^{2^{-M}} - \pi_0^{2^{M+1}} -1 + 2^{-M} \geq 0$ for all positive integers $M$, which implies that (\ref{closedformeq}) is satisfied with $a = \frac{1}{2}$. Note that $\lim_{M \rightarrow \infty} g(M) = 0$. Hence, we will show that $g(M) \geq g(M+1)$ for all $M \in \mathbb{Z}_{+}$. It is easy to check that $g(M) \geq 0$ for all $M \leq 6$ using Lemma \ref{loww}. The derivative of $g$ with respect to $M$ is $ \frac{d g(M)}{d M} = \log(\pi_0) \cdot \log(0.5) \cdot \pi_0^{2^{-M}} \cdot 2^{-M} + \log(\pi_0)\cdot\log(0.5)\cdot \pi_0^{2^{M+1}}\cdot2^{2M+1} \cdot 2^{-M} - \log(2)\cdot2^{-M}$, where $\log$ is the natural logarithm. By Lemma \ref{loww}, we have $\frac{1}{2} < \pi_0 < \frac{3}{4}$, and thus $0.19 < \log{\pi_0} \cdot \log(0.5) < 0.5.$ Since $\pi_0^{2^{-M}} \leq 1$, the first term in $\frac{d g(M)}{d M}$ is upper bounded by $\frac{1}{2} \cdot 2^{-M}$. For the second term, note that $\pi_0^8 \cdot 2 < \frac{1}{3}$. Since $2^{M+1} > 8 (2M+1)$ for all $M \geq 6$, the second term is upper bounded by $\frac{1}{2} \cdot \frac{1}{3} \cdot 2^{-M}$, and $\frac{d g(M)}{d M}$ is upper bounded by $\frac{1}{2} \cdot 2^{-M} + \frac{1}{6} \cdot 2^{-M} - \log(2) \cdot 2^{-M}$,

which is negative since $\log(2) > 2/3$.
Hence, we have shown that $g(M)$ is a decreasing function on $[6,\infty]$, which concludes the proof. 
\Halmos
\endproof

We close this section with the following conjecture:
 \begin{conjecture}\label{conjecturemonotinicity}
Assume that the agents are symmetric. Then for all $n \geq 2$, (C1) holds with $f(M)=a^M$ for all $M \in \mathbb{Z}_+$, for some $a \in [1/3,1/2]$.
\end{conjecture}

\noindent Recall that $p_{n,M}=\mathbb{P}_{\pi}(|s_1^t| \leq M)$ when there are $n$ symmetric agents, where $\pi$ is the steady-state distribution of the Markov chain $(s^t:t\geq0)$.\footnote{Because of the symmetry, in the steady-state distribution, the probability of  $-M \leq s_i^t \leq M$ equals to the probability of $-M \leq s_j^t \leq M$ for any agents $i,j \in \mathcal{A}$.} Figure \ref{fig:conjecture} shows the behavior of $p_{n,M}$ and suggests the following conjecture:

 \begin{figure}[htb!]
  \centering
  \begin{minipage}[b]{0.45\textwidth}
    \includegraphics[width=1\textwidth]{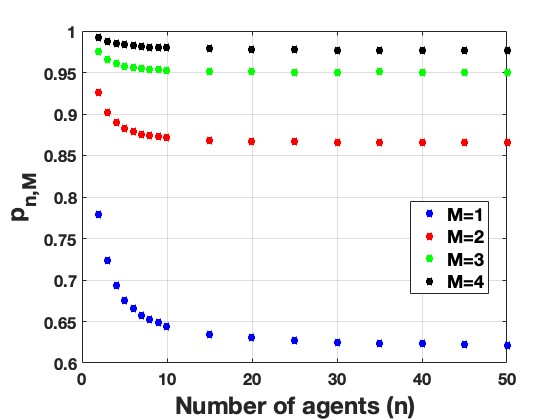}
  \end{minipage}
  \hfill
  \begin{minipage}[b]{0.45\textwidth}
    \includegraphics[width=1\textwidth]{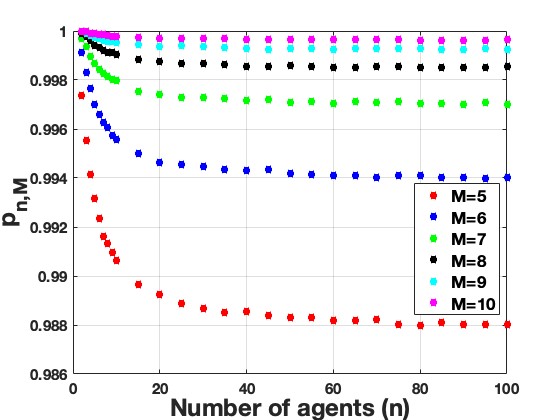}
  \end{minipage}
  \caption{The behavior of $p_{n,M}$  versus $n$. In the simulation, the market runs until $t=2\cdot10^7$ (the first $500,000$ time periods are ignored). For each $n$, the probabilities for each agent are computed separately, and the means are reported.}
  \label{fig:conjecture}
\end{figure}
 \begin{conjecture}\label{conjecture1}
$p_{{n+1},M} \leq p_{n,M}$ for all $n \geq 2$ and for all $M \in \mathbb{Z}_{+}$.
\end{conjecture}

Note that Figure \ref{fig:conjecture} suggests that for each $M$, $p_{n,M}$ converges to a limit point as $n$ grows large. The precise values from the left figure for $n=50$ are $p_{50,1} = 0.6184$, $p_{50,2} = 0.8645$, $p_{50,3} = 0.9500$ and $p_{50,4} = 0.9759$. Let $\pi_{-4} = \alpha$.  Since $p_{\infty,M} =  \pi_{-M} - \pi_{M+1}$, using (\ref{eq:rec2}), consider the following sets of equations and the corresponding positive real solutions:

$\;\;\;\;\;\;\bullet$ $p_{50,1} = \alpha^{8} - \alpha^{64} = 0.6184 \;\; \text{with two positive real roots} \;\; \alpha \approx 0.947656,\textbf{0.975067}.$

$\;\;\;\;\;\;\bullet$ $p_{50,2} =  \alpha^{4} - \alpha^{128} = 0.8645 \;\; \text{with two positive real roots} \;\; \alpha \approx 0.969503,\textbf{0.975035}.$

$\;\;\;\;\;\;\bullet$ $p_{50,3} = \alpha^{2} - \alpha^{256} = 0.9500 \;\; \text{with two positive real roots} \;\; \alpha \approx \textbf{0.975598},0.984804.$

$\;\;\;\;\;\;\bullet$ $p_{50,4} = \alpha - \alpha^{512} = 0.9759 \;\; \text{with two positive real roots} \;\; \alpha \approx \textbf{0.975904},0.991964.$

The closeness of the highlighted roots in the above sets of equations suggests that there is a consistency between the approximate limit point in Figure \ref{fig:conjecture} and our analysis for the infinite model. Note that $\pi_{-4} \approx 0.975$ implies $\pi_0 \approx 0.667$ by (\ref{eq:rec2}). Using our analysis for the infinite model and assuming Conjecture \ref{conjecture1}, observe  that (C1) is satisfied for the system with any number of symmetric agents with $\frac{1}{3} \leq a \leq \frac{1}{2}$, which implies Conjecture \ref{conjecturemonotinicity}.

\begin{figure}
  \centering
  \begin{minipage}[b]{0.45\textwidth}
    \includegraphics[width=\textwidth]{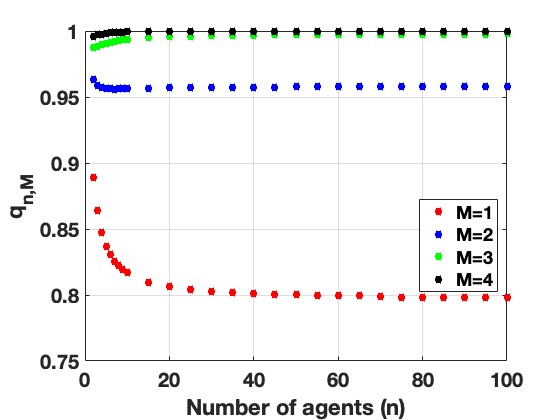}
  \end{minipage}
  \hfill
  \begin{minipage}[b]{0.45\textwidth}
    \includegraphics[width=\textwidth]{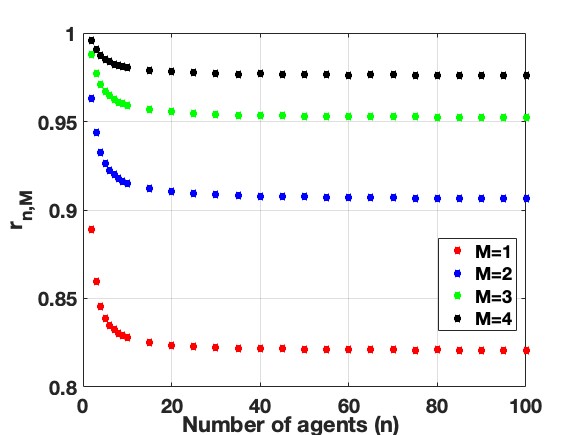}
  \end{minipage}
  \caption{The behaviors of $q_{n,M}$ and $r_{n,M}$ versus $n$. In the simulations, the market runs until $t=2\cdot10^7$ (the first $500,000$ time periods are ignored). For each $n$, the probabilities for each agent are computed separately, and the means are reported.}
  \label{fig:conjecture2}
\end{figure}

\begin{remark}
Given that there are $n$ symmetric agents,  let $q_{n,M} = \mathbb{P}(s_1^t \leq M)$ and $r_{n,M} = \mathbb{P}(s_1^t \geq -M)$, for all $M \in \mathbb{Z}_+$. Figure \ref{fig:conjecture2} shows the behaviors of $q_{n,M}$ and $r_{n.M}$. Interestingly, as Figure \ref{fig:conjecture2} shows, the monotonicity property seems to hold for $r_{n,M}$, but not for $q_{n,M}$.
\end{remark}

\section{Strategic considerations}\label{sec:game}

This section offers a microeconomic justification for the notion of stability defined by conditions (C1) and (C2). Suppose that agents gain a payoff $\pi > 0$ whenever they receive service, and incur a cost $c \in [0,\pi)$ whenever they provide service. We assume that $c<\pi$, since otherwise providing service is inefficient, and the optimal mechanism is to never provide service. 

We consider a token system $(n,P,Q,d)$ under the minimum token selection rule. For simplicity, we focus on the symmetric case of $p_i=q_i=\frac{1}{n}$ for all $i \in \mathcal{A}$, so that all agents are equally likely to request service and be available to provide service. We suppose that at time $t=0$, agents can either opt out of the token system, or else commit to participate for some fixed number of periods $T$. Assuming all $n$ agents participate, each  agent $i$ serves as a requester some (random) number of times during the process, which we denote by $A_i$, and each agent serves as a provider $B_i$ times. Thus, the total payoff to agent $i$ is $U_i := \pi A_i - c B_i$. Note that by the symmetry assumption, the expectation of $U_i$ is 
\begin{align*}
    \mathbb{E}[U_i] = \frac{T}{n}(\pi-c).
\end{align*}
In particular, this is strictly larger than zero, so if agents maximize their expected payoffs, then they will opt in. Note that this holds independent of $d$ (in particular, even in the unstable case of $d=1$). 

The advantage of stability is not that it provides agents with higher expected payoff, but rather that it lowers their risk. A simple model of risk aversion used in economics is that of mean-variance preferences.\footnote{We choose this classical model  \citep[e.g., see][]{markowitz1952portfolio, tobin1958liquidity} for its simplicity, but a similar conclusion applies for other, more sophisticated models, such as CARA preferences over monetary gambles. See \cite{pomatto2020stochastic} for a recent axiomatization of mean-variance preferences.}  Under these preferences, which are parametrized by $\kappa > 0$, agents desire a higher expected payoff but also a lower variance, and evaluate a risky prospect by subtracting from its expectation $\kappa$ times its variance. In particular, in our case agent $i$ will choose to opt in if $\mathbb{E}[U_i] - \kappa \mathrm{Var}(U_i) > 0$. 

When $d=1$, the variance of $U_i$ is linear in $T$, since $A_i$ and $B_i$ are independent. Thus, agents will only opt in if $\kappa$ is low enough. Intuitively, this unstable system is risky because agents are not unlikely to be left with many tokens at time $T$, which would mean that they provided service more often than they received it, resulting in a negative payoff.

In contrast, when $d \geq 2$, stability implies that the variance of $U_i$ is bounded from above, independent of $T$. Since the expectation of $U_i$ increases linearly, it follows that even agents with high $\kappa$  will opt in, assuming $T$ is large enough. This observation provides a foundation for why stable token systems are attractive to strategic economic agents: stable systems reduce agents' risk.


\noindent{\bf Relating to the literature.}
We further discuss our stability notion in the context  of the  literature that considers dynamic favor exchange games with discounting. Several papers study  settings with payoff $\pi$ and cost $c<\pi$  per service as above, but (instead of choosing whether to opt in) agents' strategies depend on  histories of exchanges. So an agent can refuse to provide service at any time period. In these models, each period one agent is selected at random to  request service, and another (available) agent is selected to provide service according to some selection rule.

Since $c<\pi$, always trading (provisioning a service request)  maximizes social welfare. However, demand is not always met, and there are two sources for this efficiency loss.  First, an agent may decline to provide service when her number of tokens exceeds some threshold. This indeed happens when the service provider is selected uniformly at random \citep{friedman2006efficiency, kash2015equilibrium}. The incentive to decline service is alleviated when there is a steady-state \citep{johnson2014analyzing, bo2018}, since an agent will be rewarded for her service in finite time in expectation.  Each of the conditions (C1) and (C2)  implies that the token system has a steady-state.

While the token system under the minimum token selection rule has a steady-state, an agent may still prefer to decline service provision when her number of tokens exceeds some threshold. Indeed, \citet{johnson2014analyzing} and  \citet{bo2018} show the existence of an $\epsilon$-Nash equilibrium, thus bounding  agents' benefit from declining service.\footnote{\citet{johnson2014analyzing} use punishments to attain this, whereas \citet{bo2018} assume punishments are not feasible, but they require a large market, in which many agents are available to provide service.} Our results on condition (C1) suggest how concentrated the token profiles are, and in particular, the likelihood of different level of thresholds being reached. So, for example, even if punishments are feasible, how concentrated the token profile is (i.e., the structure of $f$ in (C1)) can guide a social planner whether punishments are needed.

Second, inefficiency can arise when an agent demands a service, but has zero tokens. \cite{johnson2014analyzing} show that (when service is always provided when feasible) the minimum token selection rule minimizes the probability of this event. \citet{bo2018} show that the likelihood of this event vanishes when $n$ grows large and a constant fraction of agents are available to provide service for each request. In our setting, instead of zero,  a social planner may decide to ban an agent from  requesting service if  her tokens are below some (possibly negative) threshold. Condition (C1) implies that an agent is unlikely to have a large deficit of tokens and explicitly quantifies the probability of this event.  Finally, (C2)  complements  (C1) so that agents benefit from providing service in expectation, ensuring average positive utility in finite time.

\section{Application: Kidney Exchange} \label{sec:kidneyapplication}

In this section, we discuss how our model and insights can be applied to multi-hospital kidney exchange platforms. We begin with some background, then discuss how to apply the model, and  provide results from numerical simulations to conclude.

\subsection{Background} Kidney exchange platforms arrange swaps between incompatible patient-donor pairs so that patients receive a transplant from a compatible donor. Compatibility between a patient and a donor requires both ABO (blood-type) compatibility\footnote{The donor cannot have a blood protein that the patient does not have.} and tissue-type compatibility, which means that the patient cannot have an antibody to one of the donor's antigens. Patient sensitivity is based on her  antibodies.\footnote{The Panel Reactive Antibody (PRA)  measures  the likelihood of a patient to be tissue-type incompatible with a random donor in the population based on her antibodies.}

\noindent{\bf Easy- and hard-to-match pairs.} Pairs  can generally be categorized as easy- or hard-to-match based on their  blood types and the patient sensitivity levels. To get some intuition, assume for simplicity that there are just two blood types, A and O. In a typical pool, there are many more O-A patient-donor pairs than A-O patient-donor pairs \citep{roth2007efficient}.\footnote{An A patient who is compatible with her intended O donor will go through a direct donation.} The shortage of O donors implies that O-A pairs are hard-to-match. A-O patient-donor pairs can be easy- or hard-to-match, depending on whether the A patient is highly sensitized or not.\footnote{Note that the patient is more likely to be sensitized than a random patient as it is tissue-type incompatible with her donor.} Similarly,  X-X incompatible pairs (O-O or A-A) are ABO compatible with each other, and a majority of these pairs have highly sensitized patients making them hard-to-match pairs \citep{agarwal2019market}. 
Efficiency is aligned with matching hard-to-match with easy-to-match pairs. Intuitively, easy-to-match pairs have a positive contribution to the exchange pool while hard-to-match pairs may compete to match with easy-to-match ones.  



\noindent{\bf Free riding and token systems.}
In the US, hospitals' participation on  platforms is voluntary, and they can decide which  pairs to enroll. A common  behavior that emerged is enrolling pairs they cannot match internally \citep{agarwal2019market}. This behavior emerged since  matching rules on platforms do not account for the types of pairs each hospital enrolls.\footnote{Consider for example a ``sub-market" with  A-O and O-A pairs, and two hospitals. Suppose one hospital enrolls all pairs but the other enrolls only O-A pairs that cannot be matched internally. If the platform selects randomly for each A-O which O-A pair to match with, then the second hospital will benefit without contributing to the system.}
To alleviate free riding behavior, the National Kidney Registry (NKR)  adopted an ad-hoc token system that rewards hospitals based on their marginal contribution to the platform.\footnote{This program is called the Center Liquidity Contribution program (see \url{http://www.kidneyregistry.org/docs/CLC_Guidelines.pdf}), inspired by \citet{ashlagi2014free} who pointed out the need for a ``frequent-flyer" program.} Exchanges arranged through the platform  generate a transfer of tokens between  hospitals and the platform. The token values for each type of pair are based on the type of the pair, and the  intention is to capture the marginal benefit to the platform. Notably,  the number of tokens is negative for hard-to-match pairs (such as O-A pairs, or other pairs with highly sensitized patients), and positive for easy-to-match pairs (such as A-O or X-X  patient-donor pairs with  patients that are not highly sensitized).\footnote{Almost all platforms also reward hospitals with ending altruistic-donor chains in their hospitals based on how many of altruistic donors they submitted to the system \citep{ashlagi2013kidney}.} The actual values have  been updated  over time based on experimentation and the marginal benefit observed in practice.

\subsection{Applying the model}

Our stylized model can be applied here as follows. Hospitals can be viewed as  agents in our model and we limit our attention to exchanges that include two pairs. At each time period an easy-to-match pairs arrives to  a (random)  hospital. A random subset of  hospitals have a hard-to-match pair in the pool that can exchange with the arriving easy-to-match pair.\footnote{It is useful to consider A-O and O-A pairs as a separate economy than the one of  O-O pairs. But as discussed earlier, within each of these, there are both easy- and hard-to-match pairs.} One of these hospitals is chosen for the exchange according to a selection rule, and the chosen hospital  pays one token to the hospital with the easy-to-match pair. The maximum token selection rule selects the hospital (among the subset) with the most number of tokens.

Since for each hospital, the number of tokens it has equals the difference between the number of easy- and number of hard-to-match pairs it matched,  the maximum token selection rule favors the hospital (with a compatible hard-to-match pair) who has the largest net contribution to the platform in order to match its  hard-to-match pairs.

Observe that in our stylized model, we associate a service request with an arrival of an easy-to-match pair to a (random) hospital, and available agents correspond to a subset of hospitals that have a hard-to-match pair in the pool that is compatible with the easy-to-match pair. Similarly in our model, since the number of tokens a hospital has equals the difference between the number of hard- and number of easy-to-match pairs it matched, the minimum token selection rule favors the hospital that contributed the most to the platform. This modeling is chosen to be consistent with the literature on trading favors and the power of two choices. Note that this is without loss of generality for our purposes, since the distribution of positive number of tokens under the maximum token selection rule symmetrically mirrors the distribution of negative number of tokens under the minimum token selection rule and vice versa. 



A few comments are in place.  First, in practice, easy-to-match pairs match quickly \citep{ashlagi2018effect},  which is aligned with our assumption that such agents are served immediately.  Moreover, hard-to-match pairs whose patients are highly sensitized may have very few possible matches \citep{ashlagi2012need}, which motivates the consideration of small values of $d$.

Tracking the evolution of the underlying compatibility graph  is highly intractable. Instead, our model  abstracts  away  from the graph structure and even from counting the number of  hard-to-match pairs  hospitals have in the pool (our service availability distribution $Q$ remains fixed).  Despite this abstraction, our model can shed light on  the predicted behavior of token systems for kidney exchange platforms. 

The symmetric case corresponds further to the case where hospitals are of the same size and have the same balance between easy- and hard-to-match pairs. That these balances are similar across hospitals is a natural assumption due to the biological structure. Next, we provide numerical simulations to provide further insights for the case in which hospitals differ in size and balance. We also discuss a simple asymmetric case with two types of agents in Appendix \ref{appendix:asymmetric} and derive analogous differential equations to the symmetric case, which can be used for numerical studies.

\subsection{Simulations}\label{extensionsection2}
In this section, we present results from numerical simulations using data from the  National Kidney Registry (NKR) platform. We simulated the token distribution of participating hospitals under the minimum token selection rule. There are 1881 patient-donor pairs and 84 hospitals in the data. We restrict our attention only  to  exchanges between two patient-donor pairs (2-way cycles) and refer to such exchanges as matches. Each time period represents one day. All hospitals have 0 tokens at the beginning, and the kidney exchange pool is initially empty.  At each time period, the following steps are conducted in order:

$\bullet$ \textit{\textbf{Step 1}: Sample with replacement a  patient-donor pair $p$ uniformly at random from the entire set of pairs. This pair $p$ is a tentative service requester.}

$\bullet$ \textit{\textbf{Step 2}: Among the pairs that are waiting in the pool, identify the set of  patient-donor pairs  who can match with the pair $p$.}

$\;\;\;\;\; -$ \textit{\textbf{Step 2.1}: If $p$ has more than one possible match, use the minimum token selection rule to determine the service provider; that is,  match $p$ with the patient-donor pair that belongs to the hospital with the least amount of tokens (ties are broken uniformly at random). After the match is performed, the hospital of the service requester pays one token to the hospital of the service provider.}

$\;\;\;\;\;-$ \textit{\textbf{Step 2.2}:: If $p$ has no possible matches,  add $p$ to the pool ($p$ is no longer a tentative service requester).}

$\bullet$ \textit{\textbf{Step 3}:  Each patient-donor pair in the pool leaves the system unmatched with probability $\frac{1}{365}$, independently.}

 Observe that we do not categorize the easy- and hard-to-match pairs based on characteristics prior to the simulation.  Instead, a pair is considered as a service requester (an easy-to-match pair) if it can match upon arrival to some pair in the pool. Figure \ref{84hosp1} shows the token distribution of 84 hospitals after running the simulation $10^5$ time periods. Note that overall, the token distribution is remarkably stable.\footnote{We note that the token distribution shows a similar behavior under the following alternative specification. When a pair is sampled, instead of associating the pair with its original hospital, it is associated with a hospital randomly, where probabilities are proportional to hospital sizes.} 

\begin{figure}[]
  \centering
  \begin{minipage}[b]{0.45\textwidth}
    \includegraphics[width=1\textwidth]{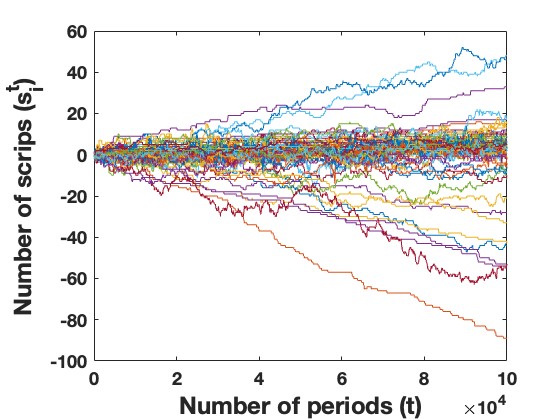}
    \caption{The token distribution of 84 hospitals under the minimum token selection rule after $10^5$ time periods.}
\label{84hosp1}
  \end{minipage}
  \hfill
  \begin{minipage}[b]{0.45\textwidth}
    \includegraphics[width=1\textwidth]{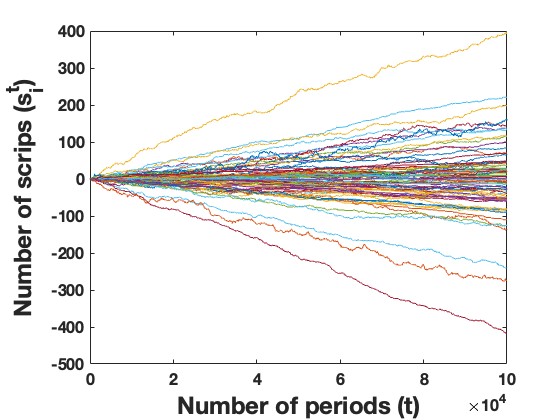}
    \caption{The token distribution of 84 hospitals under the uniform selection rule after $10^5$ time periods.}
\label{84hosp2}
  \end{minipage}
\end{figure}

The are seven hospitals whose number of tokens drop below -30; the number of pairs in these hospitals are $1,2,2,2,2,5$, and 28. The reason of deviation for the hospitals with only few pairs is apparent from the data; these hospitals only have easy-to-match  pairs, and they  match immediately upon arrival. The hospital with 28 pairs also has a large imbalance in favor of easy-to-match pairs (in contrast to the composition in typically large hospitals). This prevents the token system from rewarding back such hospitals. 
Three hospitals reach more than 30 tokens and have the following number of pairs: $6,13$, and $18$. The reason for the deviation in this case is that almost all of their pairs are hard-to-match. Figure \ref{84hosp2} shows the token distribution of hospitals when the service provider is chosen uniformly at random, instead of using the minimum token selection rule. Note that there are large deviations and no oscillations.
\begin{remark}
The simulations also reveal  that when there is at least one match in the pool for the service requester, almost 67\% of the time there are at least two compatible pairs for the service requester; this suggests that there is often multiplicity of matchings (tie-breakings). However, the  average number of matches for an easy-to-match pair (i.e., when a pair matches immediately upon arrival) is almost 7. This suggests that the average number of potential matches in the pool (or service availability)  is small. So in this case, even few ties allow the token system to be stable.
\end{remark}

\section{Final remarks}

This paper adapts methodologies from stochastic processes and ideas from the power of two choices literature to illustrate that token systems based on the minimum token selection rule are likely to behave well even in thin  marketplaces, where there can be very little  availability of supply.  We identified settings, under which the token system is stable when only few agents are available to provide service for any service request. Our analysis further provides concentration bounds for the token distribution. We further discuss why a stable token system is attractive for strategic agents interested in low participation risk.

In the context of kidney exchange, multi-hospital exchange platforms increase the chance for hospitals to match their hard-to-match pairs. To alleviate free riding by hospitals, token systems have been proposed and applied to incentivize participation by hospitals  by accounting their contribution to the platform \citep{agarwal2019market}.
Our findings suggest that the possibility of breaking ties in the matching process, even among few hospitals, enables the stability of the token system. Breaking ties based on  token balances allows platforms to reward hospitals with their contribution,  assuring  the difference between hospitals' tokens remains small. Numerical experiments reveal that tie-breakings are  likely to happen in kidney exchange,  and token systems that account for hospitals' contributions can ensure cooperation between hospitals. In practice, platforms use ad-hoc rules to increase participation, and several platforms seek to equalize the number of donors and patients of each ``player" transplanted in every  match \citep{biro2019modelling,mincu2020ip}. Our findings suggest that an arguably simpler (and less restrictive) token-based system is likely to behave well. It is interesting to expand this work to identify conditions under which  dynamic token systems can be sustainable when the market exhibits more  heterogeneity.

\bibliographystyle{informs2014}
\bibliography{scrips}


\newpage

\section{Appendix}



\subsection{Proofs from Section \ref{subsection:d=1}} \label{appendix:section2}
We start with the proof of the following proposition regarding the case $d=1$, and then discuss about the transition dynamics.

\begin{proposition}
\label{secondlemma}
The token system is not stable for any $n\geq2, P>0$ and $Q>0$ when $d=1$.
\end{proposition}

\proof{Proof of Proposition \ref{secondlemma}.}
Fix an agent $i \in \mathcal{A}$. Then $s_i^t$ is a lazy random walk in one dimension. In particular, we have
\begin{equation}\label{onewalk}
  s_i^{t+1} =  \begin{cases}
      s_i^t +1 & \text{with probability} \sum_{j \in \mathcal{A}\backslash\{i\}} p_j q_i \\
      s_i^t & \text{with probability} \;\; 1 - \sum_{j \in \mathcal{A}\backslash\{i\}} p_j q_i   - \sum_{j \in \mathcal{A}\backslash\{i\}} p_i q_j \\
      s_i^{t}-1 & \text{with probability} \sum_{j \in \mathcal{A}\backslash\{i\}} p_i q_j
   \end{cases}
\end{equation}
for all $t \geq 0$. Note that if  $\sum_{j \in \mathcal{A}\backslash\{i\}} p_j q_i \neq \sum_{j \in \mathcal{A}\backslash\{i\}} p_i q_j$, then $\mathbb{E}[s_i^t]$ diverges as $t \rightarrow \infty$ (i.e., the random walk is transient). Now assume that
\begin{equation}\label{necessary}
\sum_{j \in \mathcal{A}\backslash\{i\}} p_j q_i = \sum_{j \in \mathcal{A}\backslash\{i\}} p_i q_j \;\; \text{for all}\;\; i \in \mathcal{A}.
\end{equation}

\noindent  Then $s_i^t$ is a lazy symmetric random walk. It is a well-known fact that $s_i^t$ will take all the values in $\mathbb{Z}$ with probability $1$ (i.e., the symmetric random walk is recurrent). Moreover, even though $s_i^t$ will be $0$ infinitely often with probability $1$, the expected return time to $0$ is infinity (i.e., the symmetric random walk is null recurrent). Hence, (C1) and (C2) are not satisfied for any $n\geq2, P>0$ and $Q>0$ when $d=1$. \Halmos

As an immediate corollary to Proposition \ref{secondlemma}, consider the random tie-breaking selection rule, under which the policy selects a service provider among $d \geq 1$ available providers uniformly at random. The behavior of the random walk $s_i^t$ follows a similar structure as in \eqref{onewalk}, where the transition probabilities are state-independent;\footnote{The transition probabilities of both random walks are identical if the agents are symmetric.} unlike \eqref{twowalk}, where the transition probabilities are state-dependent and there is a possibility of stabilizing the random walk. Therefore, $s_i^t$ is again either transient or null-recurrent random walk.

\begin{corollary}\label{randomtiebreaking}
Consider a token system $(n,P,Q,d)$ under the random tie-breaking selection rule. Then the token system is not stable for any $n\geq2, P>0$ and $Q>0$ when $d\geq1$.
\end{corollary}

\endproof

\myparagraph{Transition dynamics.}  For future calculations, let us denote the following outcome by $(i,(j,k))$: $i \in \mathcal{A}$ is the service requester, and $\{j,k\} \subset \mathcal{A}$ is the subset of agents who are available to provide service. At each time period $t$, the outcome $(i,(j,k))$ occurs with probability $p_i 2q_j q_k$ if $j \neq k$ and with probability $p_i q_j^2$ if $j=k$. Hence, $i \in \mathcal{A}$ is the service requester at time $t$, $j \in \mathcal{A}$ is the service provider if $s_k^{t} > s_j^{t}$, $k \in \mathcal{A}$ is the service provider if  $s_j^{t} > s_k^{t}$, and one of $j,k \in \mathcal{A}$ is selected as the service provider uniformly at random if $s_j^{t} = s_k^{t}$. Let  $r_{jk}^t$, $j,k \in \mathcal{A}$, $j \neq k$, be the probability that given that only agents $j$ and $k$ are available to provide service at time $t$, agent $j$ is the service provider. Then, we have

\begin{equation*}
  r_{jk}^t =  \begin{cases}
      1 & \text{if} \; s_k^t > s_j^t \\
      \frac{1}{2} & \text{if} \;  s_k^t = s_j^t \\
       0 & \text{if} \;  s_k^t  < s_j^t  \\
   \end{cases}.
\end{equation*}

\noindent Fix $i \in \mathcal{A}$. We have $s_i^{t+1} = s_i^t +1$ if one of the following outcomes occurs at time $t$:

$\bullet$ $(j,(i,i))$, where $j \neq i$, which happens with probability $\sum_{j \in \mathcal{A} \backslash \{i\}} p_j q_i^2$.

$\bullet$ $(j,(i,k))$, where $j,k \neq i$, $s_k^t \geq s_i^t$ and agent $i$ wins the tiebreak if any, which happens with probability $\sum_{j \in \mathcal{A} \backslash \{i\}} \sum_{k \in \mathcal{A} \backslash \{i\}} p_j 2 q_i q_k r_{ik}^t$.

\noindent Similarly, $s_i^{t+1} = s_i^t -1$ if one of the following outcomes occurs at time $t$:

$\bullet$ $(i,(j,j))$, where $j \neq i$, which happens with probability $\sum_{j \in \mathcal{A} \backslash \{i\}} p_i q_j^2$.

$\bullet$ $(i,(j,k))$, where $j,k \neq i$, $j \neq k$, which happens with probability $\sum_{j \in \mathcal{A} \backslash \{i\}} \sum_{k \in \mathcal{A} \backslash \{i,j\}} p_i q_j q_k$.

$\bullet$ $(i,(i,j))$, where $j \neq i$, $s_i^t \geq s_j^t$ and agent $i$ loses the tiebreak if any, which happens with probability $\sum_{j \in \mathcal{A} \backslash \{i\}} p_i 2 q_i q_j r_{ji}^t$.

\noindent Therefore, we have
\begin{equation}\label{twowalk}
s_i^{t+1} =  \begin{cases}
      s_i^t +1 & \text{with probability}  \sum_{j \in \mathcal{A} \backslash \{i\}} p_j q_i^2 + \sum_{j \in \mathcal{A} \backslash \{i\}} \sum_{k \in \mathcal{A} \backslash \{i\}} p_j 2 q_i q_k r_{ik}^t\\
      s_i^{t}-1 & \text{with probability}   \; \sum_{j \in \mathcal{A} \backslash \{i\}} p_i q_j^2 +  \sum_{j \in \mathcal{A} \backslash \{i\}} \sum_{k \in \mathcal{A} \backslash \{i,j\}} p_i q_j q_k + \sum_{j \in \mathcal{A} \backslash \{i\}} p_i 2 q_i q_j r_{ji}^t\\
      s_i^t & \text{otherwise}
   \end{cases}.
\end{equation}

\noindent Note that (\ref{twowalk}) behaves similar to (\ref{onewalk}); the main difference is that the transition probabilities in (\ref{twowalk}) change as $t$ changes because of the $r_{ij}^t$ terms, which are time-dependent. Processes such as (\ref{twowalk}) referred as heterogeneous random walks, where the transition probabilities are state-dependent.

\subsection{Proofs from Section \ref{sec:2player}}\label{appendix:section3}

We start with the proof of Proposition \ref{twoagentsthm}. First, we describe the model here for convenience. Assume that there are 2 agents ($n=2$). We analyze the following discrete time birth-death process with the state space $\mathcal{S} = \left\{ (a,b) :  a \geq 1, a \in \mathbb{N}, b=1,2 \right\} \cup \left\{(0,0)\right\}$, which captures the system: the state $(a,b)$ represents the case in which agent $b$ has $a > 0$ tokens, and $(0,0)$ is the initial state. Denote this birth-death process by $(Z_t: t \geq 0)$. Let $X^t = \max_{i \in \mathcal{A}} s_i^t$ and $Y^t = \min_{i \in \mathcal{A}} s_i^t$.

\begin{repeattheorem}[Proposition \ref{twoagentsthm}.]
The token system with $2$ agents is stable if and only if $q_1^d < p_1$, $q_2^d < p_2$ and $d \geq 2$.
Let $\pi$ be the the steady-state distribution of the Markov chain $(s^t:t\geq0)$. Then for all $M \in \mathbb{Z}_{+}$, we have $$ \mathbb{P}_{\pi}(|s_i^t| \leq M)=1 -  \dfrac{\Bigg( \dfrac{p_2 q_1}{p_1 - q_1^d} \Bigg(\dfrac{p_2 q_1^d}{p_1(1-q_1^d)} \Bigg)^M  + \dfrac{p_1 q_2}{p_2 - q_2^d} \Bigg(\dfrac{p_1 q_2^d}{p_2(1-q_2^d)} \Bigg)^M \Bigg)}{\Bigg( 1 + \dfrac{p_2 q_1}{p_1 -q_1^d}  + \dfrac{p_1 q_2}{p_2-q_2^d}  \Bigg)},$$ for $i=1,2$. Moreover, the expected time between two successive occurrences of the initial state $(0,0)$ is given by $  1 + \dfrac{p_2 q_1}{p_1 -q_1^d}  + \dfrac{p_1 q_2}{p_2-q_2^d}.$

\end{repeattheorem}

\proof{Proof of Proposition \ref{twoagentsthm}.}

The transition probabilities for the process $(Z_t: t \geq 0)$ are as follows:
\begin{align*}
    &\Pr( Z_{t+1} = (1,1) \;|\; Z_{t} = (0,0)) = p_2 \bigg(\sum_{i=1}^d \dfrac{i}{d} {d \choose i} q_1^i q_2^{d-i}\bigg),\\
    &\Pr( Z_{t+1} = (0,0) \;|\; Z_{t} = (1,1)) = p_1 (1 - q_1^d),\\
    &\Pr( Z_{t+1} = (1,2) \;|\; Z_{t} = (0,0)) = p_1 \bigg(\sum_{i=1}^d \dfrac{i}{d} {d \choose i} q_2^i q_1^{d-i}\bigg),\\
    &\Pr( Z_{t+1} = (0,0) \;|\; Z_{t} = (1,2)) = p_2 (1 - q_2^d),\\
    &\Pr( Z_{t+1} = (a+1,1) \;|\; Z_{t} = (a,1)) = p_2 q_1^d \;\; \text{for all} \;\; a \geq 1,\\
    &\Pr( Z_{t+1} = (a+1,2) \;|\; Z_{t} = (a,2)) = p_1 q_2^d \;\; \text{for all} \;\; a \geq 1,\\
    &\Pr( Z_{t+1} = (a-1,1) \;|\; Z_{t} = (a,1)) = p_1 (1 - q_1^d) \;\; \text{for all} \;\; a \geq 2,\\
    &\Pr( Z_{t+1} = (a-1,2) \;|\; Z_{t} = (a,2)) = p_2 (1 - q_2^d) \;\; \text{for all} \;\; a \geq 2.
\end{align*}

\noindent Assume that the steady-state exists, and denote the steady-state vector by $\pi$. The detailed balance equations are:

\begin{equation}\label{messy1}
\pi_{(0,0)} p_2 \bigg(\sum_{i=1}^d \dfrac{i}{d} {d \choose i} q_1^i q_2^{d-i}\bigg) = \pi_{(1,1)} p_1 (1-q_1^d),
\end{equation}

\begin{equation}\label{messy2}
\pi_{(0,0)} p_1 \bigg(\sum_{i=1}^d \dfrac{i}{d} {d \choose i} q_2^i q_1^{d-i}\bigg) = \pi_{(1,2)} p_2 (1-q_2^d),
\end{equation}

\begin{equation}\label{app3:1}
\pi_{(a,1)} p_2 q_1^d = \pi_{(a+1,1)} p_1 (1-q_1^d) \;\;  \text{for all} \;\; a \geq 1,
\end{equation}

\begin{equation}\label{app3:2}
\pi_{(a,2)} p_1 q_2^d = \pi_{(a+1,2)} p_2 (1-q_2^d) \;\;  \text{for all} \;\; a \geq 1,
\end{equation}

\begin{equation}\label{app3:3}
\pi_{(0,0)} + \sum_{b=1}^2 \sum_{a=1}^{\infty} \pi_{(a,b)} = 1.
\end{equation}

\noindent It follows by (\ref{app3:1}) and (\ref{app3:2}) that

\begin{equation}\label{app3:4}
\pi_{(a,1)} = \pi_{(1,1)} \bigg( \dfrac{p_2 q_1^d}{p_1 (1-q_1^d)} \bigg)^{a-1}  \;\;  \text{for all} \;\; a \geq 1,
\end{equation}

\begin{equation}\label{app3:5}
\pi_{(a,2)} = \pi_{(1,2)} \bigg( \dfrac{p_1 q_2^d}{p_2 (1-q_2^d)} \bigg)^{a-1}  \;\;  \text{for all} \;\; a \geq 1.
\end{equation}

\noindent Using (\ref{messy1}), (\ref{messy2}), (\ref{app3:3}), (\ref{app3:4}) and (\ref{app3:5}), first note that since the infinite geometric series in (\ref{app3:3}) must converge, the following are necessary and sufficient conditions for $(Z_t: t \geq 0)$ to have a steady-state:

\begin{equation}\label{blabla1}
 \dfrac{p_2 q_1^d}{p_1 (1-q_1^d)} < 1, \; \text{or, equivalently} \; q_1^d < p_1,
\end{equation}

\begin{equation}\label{blabla2}
 \dfrac{p_1 q_2^d}{p_2 (1-q_2^d)} < 1, \; \text{or, equivalently} \; q_2^d < p_2,
\end{equation}

\noindent where the equivalences in \eqref{blabla1} and \eqref{blabla2} follow from the fact that $p_1+p_2=1$. Using (\ref{app3:3}), (\ref{blabla1}), and (\ref{blabla2}), we have

\begin{equation}\label{app3:6}
\pi_{(0,0)} + \pi_{(1,1)} \bigg( \sum_{a=1}^{\infty} \bigg( \dfrac{p_2 q_1^d}{p_1 (1-q_1^d)} \bigg)^{a-1} \bigg) + \pi_{(1,2)} \bigg( \sum_{a=1}^{\infty} \bigg( \dfrac{p_1 q_2^d}{p_2 (1-q_2^d)} \bigg)^{a-1} \bigg) = 1.
\end{equation}

\noindent Using (\ref{messy1}), (\ref{messy2}) and (\ref{app3:6}), we have

\begin{equation}\label{app3:7}
\pi_{(0,0)} = \bigg( 1 + \dfrac{p_2 (\sum_{i=1}^d \frac{i}{d} {d \choose i} q_1^i q_2^{d-i})}{p_1 - q_1^d}  + \dfrac{p_1 (\sum_{i=1}^d \frac{i}{d} {d \choose i} q_2^i q_1^{d-i})}{p_2-q_2^d}  \bigg)^{-1}.
\end{equation}

\noindent Note that $\sum_{i=1}^d \frac{i}{d} {d \choose i} q_1^i q_2^{d-i} = q_1 \sum_{i=1}^d {d-1 \choose i-1} q_1^{i-1} q_2^{d-i} = q_1 (q_1+q_2)^{d-1} = q_1$. Similarly, $\sum_{i=1}^d \frac{i}{d} {d \choose i} q_2^i q_1^{d-i}= q_2$. Hence, (\ref{app3:7}) simplifies to

\begin{equation}
\pi_{(0,0)} = \bigg( 1 + \dfrac{p_2 q_1}{p_1 -q_1^d}  + \dfrac{p_1 q_2}{p_2-q_2^d}  \bigg)^{-1}.
\end{equation}

\noindent Once we have $\pi_{(0,0)}$ as a function of $p_i$'s and $q_i$'s, we can write all the steady-state probabilities as a function of $p_i$'s and $q_i$'s. Note that in the steady-state, $X^t$ is bounded above by $M \in \mathbb{Z}_{+}$ with probability $1 - \sum_{a=M+1}^{\infty} \pi_{(a,1)} - \sum_{a=M+1}^{\infty} \pi_{(a,2)}$, which is equal to

 \begin{equation}\label{app3:10}
1 -  \dfrac{\bigg( \dfrac{p_2 q_1}{p_1 - q_1^d} \bigg(\dfrac{p_2 q_1^d}{p_1(1-q_1^d)} \bigg)^M  + \dfrac{p_1 q_2}{p_2 - q_2^d} \bigg(\dfrac{p_1 q_2^d}{p_2(1-q_2^d)} \bigg)^M \bigg)}{\bigg( 1 + \dfrac{p_2 q_1}{p_1 -q_1^d}  + \dfrac{p_1 q_2}{p_2-q_2^d}  \bigg)}.
\end{equation}

\noindent Since $X^t + Y^t = 0$ for all $t \geq 0$, for all $M \in \mathbb{Z}_{+}$, $\mathbb{P}_{\pi}(|s_i^t| \leq M)$ is given by (\ref{app3:10}) for $i=1,2$. 

It is a well-known fact that starting from a state $(a,b)$, the expected time of the first occurrence of state $(a,b)$ is $\frac{1}{\pi_{(a,b)}}$. Hence, starting from state $(0,0)$, the expected time of the first occurrence of state $(0,0)$ is

\begin{equation}\label{app3:11}
\frac{1}{\pi_{(0,0)}} = \bigg( 1 + \dfrac{p_2 q_1}{p_1 -q_1^d}  + \dfrac{p_1 q_2}{p_2-q_2^d}  \bigg).
\end{equation}
\Halmos
\endproof

\begin{remark}\label{constantaremark}

$\mathbb{E}(T_0)$ in (C2) follows from (\ref{app3:11}). The constant $a$ in (C1) can be found using (\ref{app3:10}) as follows. Define $$x = \dfrac{p_2 q_1^d}{p_1(1-q_1^d)}, y = \dfrac{p_1 q_2^d}{p_2(1-q_2^d)}, c_1 = \dfrac{\dfrac{p_2 q_1}{p_1 - q_1^d}}{1 + \dfrac{p_2 q_1}{p_1 -q_1^d}  + \dfrac{p_1 q_2}{p_2-q_2^d} }, \; \text{and} \; c_2 =  \dfrac{\dfrac{p_1 q_2}{p_2 - q_2^d}}{1 + \dfrac{p_2 q_1}{p_1 -q_1^d}  + \dfrac{p_1 q_2}{p_2-q_2^d} }.$$ Then (\ref{app3:10}) becomes $1-c_1 x^M - c_2 y^M$. We want to find $0 < a < 1$ such that $1-c_1 x^M - c_2 y^M \geq 1- a^M$, or $a^M \geq c_1 x^M + c_2 y^M$ for all $M \in \mathbb{Z}_{+}$. Note that $c_1, c_2 < 1$. Hence, $a$ can be chosen to be $x+y$ if $x+y < 1$. Consider the case when $x+y > 1$. Let $z = \max\left\{x,y\right\}$. Then $c_1 x^M + c_2 y^M < 2z^M.$ Clearly, there exists $M' \in \mathbb{Z}_{+}$ such that $2z^M < 1$ for all $M \geq M'$. Pick $0 < a_1 < 1$ such that $a_1 > z$ and $a_1^{M'}> 2 z^{M'}$. Consider the inequalities $a^M \geq c_1x^M + c_2 y^M$ for all $M < M'$. Since there are finitely many inequalities, we can pick $0 < a_2 <1$ such that these inequalities are satisfied. Thus in this case, $a$ can be chosen to be $\max\left\{a_1,a_2\right\}$.



\end{remark}


\myparagraph{Intermediate availability.} Suppose that at each time period, we have $d=2$ with probability $\beta$ and $d=1$ with probability $1-\beta$, independently, for some $\beta \in (0,1)$. We can capture this system with the same birth-death process $(Z_t: t\geq0)$  with updated transition probabilities. Following similar calculations as in the proof of Proposition \ref{twoagentsthm}, it follows that the following are necessary and sufficient conditions for the process to have a steady-state:

\begin{equation}
 \dfrac{\beta p_2 q_1^d + (1 - \beta) p_2q_1}{\beta p_1 (1 - q_1^d) + (1 - \beta)p_1q_2}  < 1,
\end{equation}

\begin{equation}
 \dfrac{\beta p_1 q_2^d + (1 - \beta)p_1q_2}{\beta p_2 (1 - q_2^d) + (1 - \beta)p_2q_1} < 1.
\end{equation}

\noindent It also follows that

\begin{equation*}
\pi_{(0,0)} = \bigg( 1 + \dfrac{p_2 q_1}{\beta(p_1 - q_1^d) + (1-\beta)(p_1q_2 - p_2q_1)}  + \dfrac{p_1 q_2}{\beta(p_2 - q_2^d) + (1-\beta)(p_2q_1 - p_1q_2)}  \bigg)^{-1},
\end{equation*}

\noindent and in the steady-state, $X^t$ is bounded above by $M \in \mathbb{Z}_{+}$ with probability

\begin{equation}\label{interremark}
1 - \pi_{(0,0)} (A+B),
\end{equation}

\noindent where

\begin{equation*}
A =  \dfrac{p_2 q_1}{\beta(p_1 - q_1^d) + (1-\beta)(p_1q_2 - p_2q_1)} \bigg(\dfrac{\beta p_2 q_1^d + (1 - \beta) p_2q_1}{\beta p_1 (1 - q_1^d) + (1 - \beta)p_1q_2} \bigg)^M,
\end{equation*}

\begin{equation*}
B =  \dfrac{p_1 q_2}{\beta(p_2 - q_2^d) + (1-\beta)(p_2q_1 - p_1q_2)} \bigg(\dfrac{\beta p_1 q_2^d + (1 - \beta) p_1q_2}{\beta p_2 (1 - q_2^d) + (1 - \beta)p_2q_1} \bigg)^M.
\end{equation*}

\noindent Similarly, since $X^t + Y^t = 0$ for all $t \geq 0$, for all $M \in \mathbb{Z}_{+}$, $\mathbb{P}_{\pi}(|s_i^t|\leq M)$ is given by (\ref{interremark}) for $i=1,2$. When agents are symmetric and $d=2$, (\ref{interremark}) becomes $1-\frac{2}{2+\beta} (\frac{2-\beta}{2+\beta})^M$. As Figure \ref{fig:inter} shows, even for small values of $\beta$, the token distribution is fairly balanced with high probability.
\begin{figure}[h]
\begin{center}
\includegraphics[scale=0.45]{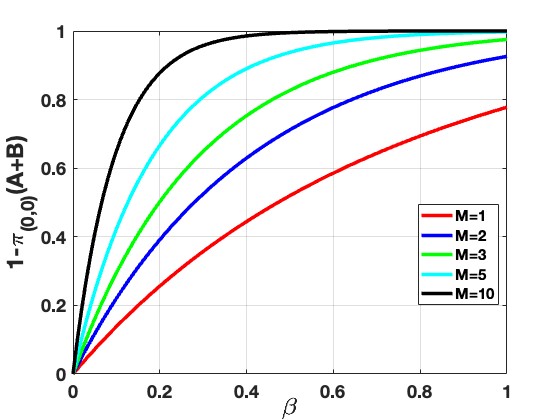}
\caption{(\ref{interremark}) as a function of $\beta$ for several values of $M$ when the agents are symmetric and $d=2$.}
\label{fig:inter}
\end{center}
\end{figure}

\subsection{Proof of Theorem \ref{thm2}}\label{appendixthm1proof}

\proof{Proof of Theorem \ref{thm2}.} 

Denote by $k^t$ the agent chosen to request service at time $t$, and denote by $I^t$ the set of agents chosen to be available to provide service at time $t$. Let $|I^t|$ denote the size of $I^t$, and note that $|I^t|$ takes values in  $\{1,2,\ldots,d\}$.  Let $j^{t} \in I^t$ be the agent chosen to provide service, i.e., an agent chosen uniformly from the agents in $i \in I^t$ that minimizes $s_i^t$. Hence $s_{j^{t}}^t = \min_{i \in I^t}s_i^t$, so that in particular, the service provider $j^{t}$ is an agent with the minimum number of tokens among the available agents in $I^t$.  

Let $V^t = \sum_{i=1}^n (s_i^t)^2$,
 and let 
    $$v^t := \mathbb{E}[V^t] = \sum_{i=1}^n \mathbb{E}[(s_i^t)^2].
    $$ 
    Note that by symmetry, we have $v^t = n\mathbb{E}[(s_1^t)^2].$  Let $E_t$ be the event $\{k^t = j^t\}$, i.e., the event that the service provider and requester are the same agent. Let $E_t^c$ be the event $\{k^t \neq j^t\}$.  Since $k^t$ is uniformly distributed and $k^t$ and $j^t$ are independent, the probability of $E_t$ is $1/n$.
    Note that conditioned on $E_{t}$, it holds that $s^{t+1}=s^t$. Then we have 
\begin{align*}
    v^{t+1} 
    &= \frac{1}{n}\mathbb{E}[V^{t+1}|E_t] + \frac{n-1}{n}\mathbb{E}[V^{t+1}|E_t^c]\\
    &= \frac{1}{n}\mathbb{E}[V^t] + \frac{n-1}{n}\mathbb{E}[V^{t+1}|E_t^c]\\
    &= \frac{1}{n}v^t + \frac{n-1}{n}\mathbb{E}[V^{t+1}|E_t^c].
    \end{align*}
 Now,
    \begin{align*}
    \mathbb{E}[V^{t+1}|E_t^c]
    &=
    \sum_{i \in \mathcal{A}} \mathbb{E}[(s_i^{t+1})^2|E_t^c]\\
    &=  \mathbb{E}[(s_{k^t}^{t+1})^2|E_t^c]+\mathbb{E}[(s_{j^{t}}^{t+1})^2|E_t^c]+ \sum_{i \in I^t \setminus \{j^t\}} \mathbb{E}[(s_i^{t+1})^2|E_t^c] + \sum_{i \in \mathcal{A} \setminus I^t \setminus \{k^t\}} \mathbb{E}[(s_i^{t+1})^2|E_t^c] \\
    &= \mathbb{E}[(s_{k^t}^{t}-1)^2 +(s_{j^{t}}^{t} + 1)^2 |E_t^c]+ \sum_{i \in I^t \setminus \{j^t\}} \mathbb{E}[(s_i^{t})^2] + \sum_{i \in \mathcal{A} \setminus I^t \setminus \{k^t\}} \mathbb{E}[(s_i^{t})^2]  \\
    &= - 2\mathbb{E}[s_{k^t}^{t}|E_t^c] + 1 + 2\mathbb{E}[s_{j^{t}}^{t}|E_t^c]+1 + \sum_{i \in \mathcal{A} } \mathbb{E}[(s_i^{t})^2].
\end{align*}


\noindent Since $\sum_{i\in \mathcal{A}} s_i^t = 0$ for all $t \geq 0$, we must have that $\mathbb{E}[s_{k^t}^{t}|E_t^c] \geq 0$. Since $k^t$ and $j^t$ are independent, we have $\mathbb{E}[s_{j^{t}}^{t}|E_t^c]=\mathbb{E}[s_{j^{t}}^{t}]$. Hence,
\begin{align*}
    \mathbb{E}[V^{t+1}|E_t^c] \leq v^t + 2\mathbb{E}[s_{j^t}^t]+2,
\end{align*}
and so
\begin{align*}
    v^{t+1} \leq v^t + \frac{2(n-1)}{n}\mathbb{E}[s_{j^t}^t]+\frac{2(n-1)}{n}.
\end{align*}



Without loss of generality, assume that ${I}^t= \{1,2,\ldots, |I^t|\}$.
Now with probability $n^{-(d-1)}$, we have that $|I^t|=1$, in which case $I^t =\{j^t\}=\{1\}$. Hence $\mathbb{E}[s_{j^t}^t\;|\; |I^t|=1]=0$. With probability $1-n^{-(d-1)}$, we have that $|I_t|\geq 2$, in which case $s_{j^t}^t=\min\{s_1^t,s_2^t, \ldots, s_{|I^t|}^t \} \leq \min\{s_1^t,s_2^t \}$. Using the fact that $\min\{a,b\} = \frac{1}{2}(a+b - |a-b|)$ for all $a,b \in \mathbb{R}$, we get $2\mathbb{E}[\min\{s_1^t,s_2^t \}]= -\mathbb{E}[|s_1^t - s_2^t |]$, and we get 
\begin{align*}
    \mathbb{E}[s_{j^t}^t \;|\; |I^t| \geq 2] \leq -\frac{1}{2}\mathbb{E}[|s_1^t-s_2^t|].
\end{align*}
Thus,
\begin{align*}
    \mathbb{E}[s_{j^t}^t \;|\; |I^t| \geq 2] \leq -\frac{1}{2}(1-n^{-(d-1)})\mathbb{E}[|s_1^t-s_2^t|],
\end{align*}
and we get
\begin{align*}
    v^{t+1} 
    &\leq v^t - \frac{n-1}{n}  (1-n^{-(d-1)})\mathbb{E}[|s_1^t - s_2^t|] + 2.
\end{align*}


Now we make use of the following claim. For real random variables $Y,Z_1,\ldots, Z_n$, we have 

\begin{equation}\label{thm2claim}
    \mathbb{E}\bigg[\bigg| Y - \frac{1}{n} \sum_{i=1}^n Z_i\bigg|\bigg] \leq \frac{1}{n} \sum_{i=1}^n \mathbb{E}[|Y - Z_i|],
\end{equation}
where the claim follows immediately from convexity and Jensen's inequality. Again by the symmetry of the problem, we have

\begin{equation*}
    \mathbb{E}[|s_1^t - s_2^t |] = \frac{1}{n-1} \sum_{i=2}^n \mathbb{E}[|s_1^t - s_i^t|]. 
\end{equation*}
Therefore by \eqref{thm2claim} and using the fact that $\sum_{i=2}^n s_i^t = -s_1^t$ for all $t \geq 0$, we have 

\begin{equation*}
    \mathbb{E}[|s_1^t - s_2^t|] \geq  \mathbb{E}\bigg[\bigg| s_1^t - \frac{1}{n-1}\sum_{i=2}^n s_i^t \bigg|\bigg] = \mathbb{E}\bigg[\bigg| s_1^t + \frac{s_1^t}{n-1} \bigg|\bigg] = \frac{n}{n-1} \mathbb{E}[|s_1^t|],
\end{equation*}
Therefore, we have

\begin{equation*}
    v^{t+1}  \leq v_t - (1-n^{-(d-1)})\mathbb{E}[|s_1^t|] + 2,
\end{equation*}
and via recursion, we get
\begin{equation*}
    v^{t} \leq 2t - (1-n^{-(d-1)})\sum_{u=0}^{t-1} \mathbb{E}[|s_1^u|].
\end{equation*}
Since $v^t \geq 0$, we have
\begin{equation}\label{thm2eq}
\frac{1}{t} \sum_{u=0}^{t-1} \mathbb{E}[|s_1^u|] \leq \frac{2}{1-n^{-(d-1)}} \leq 4,
\end{equation}
where the second inequality follows from the fact that for $d,n \geq 2$, the denominator $1-n^{-(d-1)}$ is at least $1/2$.

Denote the distribution of $s^t$ by $\mu^t \in \Delta(\mathbb{Z}^n)$, and let $\nu^t = \frac{1}{t} \sum_{u=0}^{t-1} \mu_t$. Let $Y^t$ be a random variable with distribution $\nu^t$ for all $t \geq 1$, and denote by $Y_i^t$ the $i$'th index of $Y^t$ for all $i \in \mathcal{A}$. Then per \eqref{thm2eq}, we have

\begin{equation}\label{eq:upperbound}
    \mathbb{E}[|Y_i^t|] = \frac{1}{t} \sum_{u=0}^{t-1} \mathbb{E}[|s_i^u|] \leq 4,
\end{equation}
 for all $i \in \mathcal{A}$. 

Suppose towards a contradiction that the Markov chain has no stationary probability measure. Then for any finite subset $E \subset \mathbb{Z}^n$, it holds that $\lim_{t \rightarrow \infty}\mathbb{P}(s_i^t \in E) =0$ by standard arguments about Markov chains. In particular, $\lim_{u \rightarrow \infty} \mathbb{E}[|s_i^u|] = \infty$, which contradicts  \eqref{eq:upperbound}. Thus, we have stability.

Since the Markov chain has a stationary distribution, and since it is irreducible and aperiodic, the distribution of $s_i^t$ will converge to the stationary distribution. It follows from \eqref{eq:upperbound} that $\lim_{t \rightarrow \infty} \mathbb{E}[|s_i^t|] \leq 4$, and so $\mathbb{E}[|s_i^t|] \leq 5$ for all $t$ large enough. Hence,  it follows from Markov's inequality that for all $t$ large enough and for all $i \in \mathcal{A}$, we have
\begin{align*}
    \mathbb{P}(|s_i^t| \geq M)\leq \frac{5}{M}. 
\end{align*} 
\hfill\Halmos

\subsection{Proof of Theorem~\ref{expbound}}\label{proofexpbound}


Recall that for $\lambda \in \mathbb{R}$, $Z_\lambda^t = \sum_{i=1}^n e^{\lambda s_i^t}$ and $E_\lambda^t = \mathbb{E}[Z_\lambda^t]$. As a first step, we show that for $\lambda$ sufficiently close to zero, $E_\lambda^t$ has a negative drift and thus a finite limit.

\begin{lemma}[Negative drift]\label{lem:neg-drift}
    For any $\lambda \in \mathbb{R}$, 
    \begin{equation}\label{eqn:lem-neg-drift-main}
        E_\lambda^{t+1} \le C_\lambda + (1-\gamma_\lambda) E_\lambda^t,
    \end{equation}
    where $C_\lambda := \frac{|1-e^\lambda|}{n}$ and $\gamma_\lambda := \frac{|1-e^\lambda|}{n^2} + \frac{2 - 2\cosh\lambda}{n}$. In particular, when $|\lambda| \lesssim 1/n$, we have $\gamma_\lambda > 0$ and thus, $\limsup_{t\to\infty} E_\lambda^t \le C_\lambda / \gamma_\lambda < \infty$.
\end{lemma}

\proof{Proof of Lemma~\ref{lem:neg-drift}.}
    Throughout the analysis, we will assume without loss of generality that $s_1^t \le s_2^t \le \cdots \le s_n^t$. As in the proof of Theorem~\ref{thm2}, we denote by $k^t$ and $j^t$ the agent chosen to request service and the one chosen to provide service at time $t$, respectively.
    
    By the law of total probability, we have
    \begin{equation}\label{eqn:proof_lem_bdd_growth_ltp}
        E_\lambda^{t+1} = \mathbb E\big[ \mathbb E[Z_\lambda^t|s^t]\big].
    \end{equation}
    
    First, consider $\lambda > 0$.
    Note that, through some straightforward albeit tedious arithmetic manipulation, we have
    \begin{align}
        \mathbb E[Z_\lambda^{t+1} | s^t]
        &= Z_\lambda^t + \mathbb E\Big[\big(e^{\lambda (s_{j^t}^t+1)} - e^{\lambda s_{j^t}^t}+e^{\lambda (s_{k^t}^t-1)}-e^{\lambda s_{k^t}^t}\big) \cdot \mathbbm{1}_{j^t\ne k^t}\Big|s^t\Big] \nonumber \\
        &= Z_\lambda^t + (e^\lambda - 1)\mathbb E\Big[\big(e^{\lambda s_{j^t}^t}-e^{\lambda (s_{k^t}^t-1)}\big) \cdot \mathbbm{1}_{j^t\ne k^t}\Big|s^t\Big] \\ 
        &= Z_\lambda^t + (e^\lambda - 1)\mathbb E\Big[\big(e^{\lambda s_{j^t}^t}-e^{\lambda (s_{k^t}^t-1)}\big) \Big|s^t\Big] - \nonumber \\
        &\qquad (e^\lambda - 1)\mathbb E\Big[\big(e^{\lambda s_{j^t}^t}-e^{\lambda (s_{k^t}^t-1)}\big) \cdot \mathbbm{1}_{j^t = k^t}\Big|s^t\Big] \nonumber \\
        &= Z_\lambda^t + (e^\lambda - 1)\mathbb E[e^{\lambda s_{j^t}^t} |s^t] - (1-e^{-\lambda})\mathbb E[e^{\lambda s_{k^t}^t}|s^t] - \nonumber \\
        &\qquad (e^\lambda - 1)\mathbb E\Big[e^{\lambda s_{j^t}^t} \cdot \mathbbm{1}_{j^t = k^t}\Big|s^t\Big] + (1-e^{-\lambda})\mathbb E\Big[e^{\lambda s_{j^t}^t} \cdot \mathbbm{1}_{j^t = k^t}\Big|s^t\Big] \nonumber \\
        &= Z_\lambda^t + (e^\lambda - 1)\mathbb E[e^{\lambda s_{j^t}^t} |s^t] - \frac{1-e^{-\lambda}}{n}Z_\lambda^t + \nonumber \\
        &\qquad \frac{2 - e^\lambda - e^{-\lambda}}{n}\mathbb E\Big[e^{\lambda s_{j^t}^t}\Big|s^t\Big] \nonumber \\
        &\le \Big(1 - \frac{1-e^{-\lambda}}{n}\Big)Z_\lambda^t + (e^\lambda - 1)\mathbb E[e^{\lambda s_{j^t}^t} |s^t], \label{eqn:proof_lem_drift_Etplus1_bound}
    \end{align}
    where in the second last equality we use the fact that $k^t$ is uniform and independent of $j^t$, and in the last step we drop the last term as it is non-positive.

    \begin{claim}\label{claim:diff_E}
        When $d\ge 2$, for $\lambda > 0$,
        \begin{equation}
            \mathbb E[e^{\lambda s_{j^t}^t} |s^t] - \mathbb E[e^{\lambda s_{k^t}^t} |s^t] \le \frac{1}{n^2}\sum_{1\le i < j\le n} \left( e^{\lambda s_i^t} - e^{\lambda s_j^t} \right) \le \frac{1}{n} - \frac{1}{n^2}Z_\lambda^t \le 0.
        \end{equation}
        Similarly, for $\lambda < 0$,
        \begin{equation}
            \mathbb E[e^{\lambda s_{j^t}^t} |s^t] - \mathbb E[e^{\lambda s_{k^t}^t} |s^t] \ge \frac{1}{n^2}\sum_{i=1}^n \left( e^{\lambda s_i^t} - e^{\lambda s_n^t} \right) \ge \frac{1}{n^2}Z_\lambda^t - \frac{1}{n} \ge 0.
        \end{equation}
    \end{claim}

    As a result of the claim, for $\lambda > 0$,
    \begin{align}
        \mathbb E[Z_\lambda^{t+1}|s^t] &\le \Big(1 - \frac{1-e^{-\lambda}}{n}\Big)Z_\lambda^t + \nonumber \\
        &\qquad (e^\lambda - 1) \Big( \mathbb E[e^{\lambda s_{k^t}^t} |s^t] + \frac{1}{n} - \frac{1}{n^2}Z_\lambda^t \Big) \nonumber \\
        &= \Big(1 - \frac{e^\lambda - 1}{n}\big(e^{-\lambda} - 1 + \frac{1}{n}\big)\Big)Z_\lambda^t + \frac{e^\lambda - 1}{n}. \nonumber
    \end{align}
    By the law of total probability, 
    \begin{equation}\label{eqn:proof_lem_neg_drift_rec_1}
        E_\lambda^{t+1} \le \frac{e^\lambda - 1}{n} + \Big(1 - \frac{e^\lambda - 1}{n}\big(e^{-\lambda} - 1 + \frac{1}{n}\big)\Big) E_\lambda^t.
    \end{equation}
    For $\lambda$ sufficiently small (namely, $\lambda \lesssim 1/n$), the coefficient $\gamma_\lambda := \frac{e^\lambda - 1}{n}\big(e^{-\lambda} - 1 + \frac{1}{n}\big) > 0$.

    Similarly, for $\lambda < 0$,
    \begin{align}
        \mathbb E[Z_\lambda^{t+1}|s^t] &\le \Big(1 - \frac{1-e^{-\lambda}}{n}\Big)Z_\lambda^t + \nonumber \\
        &\qquad (e^\lambda - 1) \Big( \mathbb E[e^{\lambda s_{k^t}^t} |s^t] + \frac{1}{n^2}Z_\lambda^t -  \frac{1}{n}\Big) \nonumber \\
        &= \Big(1 - \frac{1 - e^\lambda}{n}\big(1 - e^{-\lambda} + \frac{1}{n}\big)\Big)Z_\lambda^t + \frac{1 - e^\lambda}{n}. \nonumber
    \end{align}
    Hence,
    \begin{equation}\label{eqn:proof_lem_neg_drift_rec_2}
        E_\lambda^{t+1} \le \frac{1 - e^\lambda}{n} + \Big(1 - \frac{1  -e^\lambda}{n}\big(1 - e^{-\lambda} + \frac{1}{n}\big)\Big) E_\lambda^t.
    \end{equation}
    Again, the recurrent coefficient is less than 1 when $|\lambda| \lesssim 1/n$, finishing our proof.
\Halmos \endproof

\proof{Proof of Claim~\ref{claim:diff_E}.}
    Since
    \begin{equation}
        s^t_{j_t} = \min_{1\le \ell \le d} s^t_{I^t_\ell} \quad\text{ with } I^t_1,\ldots, I^t_d \overset{i.i.d.}{\sim} \text{Uniform}(\{1,\ldots,n\}),
    \end{equation}
    the distribution of $e^{\lambda s^t_{j_t}}$ is dominated by (resp. dominates) the case where $d = 2$ for $\lambda > 0$ (resp. $\lambda < 0$). It is sufficient to consider $d = 2$. Again, we assumed without loss of generality that the indices are ordered such that at time $t$ we have $s^t_1 \le s^t_2 \le \cdots \le s^t_n$. The difference in the two expectations can be evaluated as
    \begin{align}
        \mathbb E[e^{\lambda s^t_{j_t}} |s^t] - \mathbb E[e^{\lambda s^t_{k_t}} |s^t] &= \frac{1}{n^2}\sum_{1\le i < j\le n} \left( e^{\lambda s^t_i} - e^{\lambda s^t_j} \right) \nonumber \\
        &\le \frac{1}{n^2}\sum_{i=1}^n \left(e^{\lambda s^t_1} - e^{\lambda s^t_i}\right) \nonumber \\
        &= \frac{1}{n} - \frac{1}{n^2}Z_\lambda^t.
    \end{align}
    The case for $\lambda < 0$ is analogous.
\Halmos\endproof

\begin{lemma}\label{lem:chernoff-refined}
For $\lambda \in (-\ln 1.5, \ln 2)$,
    $\limsup_{t\to\infty} E_\lambda^t < \infty$.
\end{lemma}

\proof{Proof of Lemma~\ref{lem:chernoff-refined}.}
    Define $\overline \lambda$ as
    \begin{equation}
        \overline\lambda := \sup\{\lambda: \limsup_{t\to\infty} E_\lambda^t < \infty\},
    \end{equation}
    and similarly define $\underline\lambda$ as
    \begin{equation}
        \underline\lambda := \inf\{\lambda: \limsup_{t\to\infty} E_\lambda^t < \infty\}.
    \end{equation}
    Lemma~\ref{lem:neg-drift} implies that $\underline\lambda < 0 < \overline\lambda$ for any $n\in\mathbb Z_+$. Note that for any $n,t\in\mathbb Z_+$, the function $\lambda\mapsto E_\lambda^t$ is convex. 
    As a result, for any $\lambda\in(\underline\lambda(n),\overline\lambda(n))$, we have $\limsup_{t\to\infty} E_\lambda^t < \infty$. It now suffices to show that $\overline\lambda \ge \log 2$ and $\underline\lambda \le -\log 1.5$.
    
    For the moment, let us assume that $n$ is even for convenience, and write $n=2L$ with $L \ge 2$; generalizing this is straightforward.
    
    Consider $\lambda_0\in(\underline\lambda(n), 0)$. We will investigate for what values of $\lambda > 0$ we have a finite $\limsup_{t\to\infty} E_\lambda^t$.
    
    Again, we consider the recursive relationship between $E_\lambda^t$ and $E_\lambda^{t+1}$, and assume without loss of generality that $s_1^t\le\cdots\le s_n^t$. Let $\epsilon \in (0,1)$ be a constant to be specified, and denote by $\mathcal{E}_t$ the event that $e^{\lambda s_L^t}< \epsilon Z_\lambda^t/n$. We then decompose $E_\lambda(t+1)$ into the combined contribution from two terms as
    \begin{align}
        E_\lambda(t+1) &= \mathbb E[Z_\lambda(t+1)] \nonumber \\
        &= \mathbb E\big[Z_\lambda(t+1)|\mathcal{E}_t\big] + \mathbb E\big[Z_\lambda(t+1)|\mathcal{E}_t^c\big]. \label{Eqn:proof_refined_chernoff_lambda_gte_zero_two_cases}
    \end{align}
    
    For the first term in \eqref{Eqn:proof_refined_chernoff_lambda_gte_zero_two_cases}, note that whenever $e^{\lambda s_L^t}< \epsilon Z_\lambda^t/n$, we have
    \begin{align}
        \sum_{1\le i < j\le n} e^{\lambda s_i^t} - e^{\lambda s_j^t} &\le \sum_{i=1}^L \sum_{j=L+1}^n e^{\lambda s_i^t} - e^{\lambda s_j^t} \\
        &\le L(n-L) \cdot Z_\lambda^t \cdot \left(\frac{\epsilon}{n} -  \frac{1}{n-L}\left(1-\frac{\epsilon L}{n}\right)\right)\\
        &= -(1-\epsilon) L Z_\lambda^t.
    \end{align}
    Thus, in this case, Claim~\ref{claim:diff_E} can be refined into
    \begin{equation}
        \mathbb E\big[e^{\lambda s_{j^t}^t} - e^{\lambda s_{k^t}^t} \big| \mathcal{E}_t, s^t\big] \le \frac{1}{n^2}\sum_{1\le i < j\le n} e^{\lambda s_i^t} - e^{\lambda s_j^t} \le -\frac{1-\epsilon}{n^2} L Z_\lambda^t.
    \end{equation}
    Similar to \eqref{eqn:proof_lem_drift_Etplus1_bound}, we now have
    \begin{align}
        \mathbb E\Big[Z_\lambda^{t+1} \Big| \mathcal{E}_t, s^t\Big] 
        &\le \Big(1 - \frac{1-e^{-\lambda}}{n}\Big)Z_\lambda^t +  (e^\lambda - 1) \mathbb E\Big[e^{\lambda s_{j^t}^t} \Big| \mathcal{E}_t, s^t\Big] \nonumber \\
        &\le \Big(1 - \frac{1-e^{-\lambda}}{n}\Big)Z_\lambda^t + 
        (e^\lambda - 1) \Big( \mathbb E\Big[e^{\lambda s_{k^t}^t} \Big| \mathcal{E}_t, s^t\Big] - \frac{1-\epsilon}{n^2}LZ_\lambda^t \Big) \nonumber \\
        &= \bigg(1 - \frac{e^\lambda - 1}{n}\Big(e^{-\lambda} - 1 + \frac{(1-\epsilon)L}{n}\Big)\bigg)Z_\lambda^t. \nonumber
    \end{align}
    Thus,
    \begin{equation}\label{Eqn:proof_refined_chernoff_lambda_gte_zero_part_1}
        \mathbb E\big[Z_\lambda^{t+1}|\mathcal{E}_t \big] \le \bigg(1 - \frac{e^\lambda - 1}{n}\Big(e^{-\lambda} - 1 + \frac{(1-\epsilon)L}{n}\Big)\bigg) E_\lambda^t.
    \end{equation}
    For the choice of $L=n/2$, the coefficient is strictly less than 1 for any $\lambda < \log 2$ granted that we choose $\epsilon$ sufficiently small.
    
    Now we consider the second term in \eqref{Eqn:proof_refined_chernoff_lambda_gte_zero_two_cases} under the event $\mathcal{E}_t^c$. 
    By a natural extension of Lemma~\ref{lem:neg-drift}, we have
    \begin{equation*}
        \mathbb E\big[Z_\lambda^{t+1}|\mathcal{E}_t^c \big] \le \mathbb E\Big[\big(C_\lambda + (1-\gamma_\lambda) Z_\lambda^t\big) \cdot\mathbbm{1}_{e^{\lambda s_L^t} \ge \epsilon Z_\lambda^t/n} \Big]
        \le C_\lambda + (1-\gamma_\lambda) \mathbb E\Big[Z_\lambda^t \cdot\mathbbm{1}_{e^{\lambda s_L^t} \ge \epsilon Z_\lambda^t/n}\Big].
    \end{equation*}
    The expectation term on the right-hand side can be further expanded as
    \begin{multline*}
        \mathbb E\Big[Z_\lambda^t \cdot\mathbbm{1}_{e^{\lambda s_L^t} \ge \epsilon Z_\lambda^t/n}\Big] = \int_0^\infty \mathbb P\big(Z_\lambda^t \cdot\mathbbm{1}_{e^{\lambda s_L^t} \ge \epsilon Z_\lambda^t/n} \ge x\big) dx \le n + \int_n^\infty \mathbb P\Big(Z_\lambda^t \ge x, e^{\lambda s_L^t} \ge \frac{\epsilon}{n} Z_\lambda^t\Big) dx.
    \end{multline*}
    When $e^{\lambda s_L^t} \ge \epsilon Z_\lambda^t/n$, we must have
    \begin{equation}
        \prod_{i=1}^{L-1} e^{-\lambda s_i^t} = \prod_{i=L}^n e^{\lambda s_i^t} \ge \bigg(\frac{\epsilon}{n} Z_\lambda^t\bigg)^{n-L+1},
    \end{equation}
    or equivalently for $\lambda_- \in(\underline\lambda, 0)$,
    \begin{equation}
        \prod_{i=1}^{L-1} e^{\lambda_- s_i^t} \ge \bigg(\frac{\epsilon}{n} Z_\lambda^t\bigg)^{\frac{-\lambda_-(n-L+1)}{\lambda}}.
    \end{equation}
    By the AM-GM inequality,
    \begin{equation}
        Z_{\lambda_-}(t) \ge \sum_{i=1}^{L-1} e^{\lambda_- s_i^t} \ge (L-1)\bigg(\frac{\epsilon}{n} Z_\lambda^t\bigg)^{\frac{-\lambda_-(n-L+1)}{\lambda (L-1)}}.
    \end{equation}
    Putting things together, we obtain
    \begin{align}
        \mathbb E\big[Z_\lambda(t+1);\mathcal{E}_t^c \big] 
        &\le C_\lambda + (1-\gamma_\lambda)  \left(n + \int_n^\infty \mathbb P\Bigg(Z_{\lambda_-}(t) \ge (L-1)\bigg(\frac{\epsilon}{n} x\bigg)^{\frac{-\lambda_-(n-L+1)}{\lambda (L-1)}}\Bigg) dx \right) \nonumber \\
        &\le C_\lambda + (1-\gamma_\lambda) \bigg(n + \int_n^\infty \frac{E_{\lambda_-}(t)}{L-1} \cdot \Big(\frac{\epsilon}{n}\Big)^{\frac{\lambda_-(n-L+1)}{\lambda (L-1)}} \cdot x^{\frac{\lambda_-(n-L+1)}{\lambda (L-1)}} dx \bigg) \nonumber \\
        &= C_\lambda + (1-\gamma_\lambda) \bigg(n + \frac{E_{\lambda_-}(t)}{L-1} \cdot \Big(\frac{\epsilon}{n}\Big)^{\frac{\lambda_-(n-L+1)}{\lambda (L-1)}} \int_n^\infty  x^{\frac{\lambda_-(n-L+1)}{\lambda (L-1)}} dx \bigg), \label{Eqn:proof_refined_chernoff_lambda_gte_zero_part_2}
    \end{align}
    where the last inequality uses Markov's inequality. The entire coefficient in front of the integral is uniformly bounded in $t$ since $\limsup_{\tau\to\infty} E_{\lambda_-}^{\tau} < \infty$; whenever $\lambda < -\lambda_-\frac{n-L+1}{L-1}$, the integral converges.
    Combining \eqref{Eqn:proof_refined_chernoff_lambda_gte_zero_part_1} with \eqref{Eqn:proof_refined_chernoff_lambda_gte_zero_part_2}, we have
    \begin{equation*}
        E_\lambda^{t+1} \le (1-\gamma') E_\lambda^t + C'
    \end{equation*}
    for some positive constants $\gamma'$ and $C'$, whenever $\lambda < \min\{-\frac{n-L+1}{L-1}\lambda_-, \log 2\}$. Thus, $\limsup_{t\to\infty} E_\lambda^t < \infty$. 
    Sending $\lambda_-\to\underline\lambda$ gives
    \begin{equation}\label{eqn:proof_refine_chernoff_final_part_1}
        \overline\lambda \ge \min\Big\{-\frac{n-L+1}{L-1}\underline\lambda, \log 2\Big\}.
    \end{equation}
    
    With an analogous argument by considering $\lambda_+\in(0,\overline\lambda(n))$ and $\lambda < 0$, we obtain
    \begin{equation*}
        E_\lambda(t+1) \le (1-\gamma'') E_t + C''
    \end{equation*}
    for some $\gamma'',C'' > 0$,
    granted that $\lambda > \max\{-\frac{n-L+1}{L-1}\lambda_+, \log\frac{2}{3}\}$. Sending $\lambda_+\to\overline\lambda$ gives
    \begin{equation}\label{eqn:proof_refine_chernoff_final_part_2}
        \underline\lambda \le \max\Big\{-\frac{n-L+1}{L-1}\overline\lambda, \log\frac{2}{3}\Big\}.
    \end{equation}
    
    \eqref{eqn:proof_refine_chernoff_final_part_1} and \eqref{eqn:proof_refine_chernoff_final_part_2} together form a system of two piecewise linear inequalities. Solving it gives
    \begin{equation*}
        \overline\lambda \ge \log 2 \qquad\text{ and }\qquad \underline\lambda \le \log\frac{2}{3}
    \end{equation*}
    as desired.
\Halmos \endproof

\proof{Proof of Theorem~\ref{expbound}.}
The theorem follows directly from Lemmas~\ref{lemma1:lyapunov} and \ref{lem:chernoff-refined}.
\Halmos\endproof

\subsection{Proofs from Section \ref{sec:general}} \label{appendix:section4}

\proof{Proof of Lemma \ref{loww}.}
We first prove that $\pi_0 > \frac{1}{2}$. Note that (\ref{eq:det2}) can be written as
\begin{equation}\label{lemma3rec}
\sum_{M \geq 1} ( \pi_M + \pi_{-M+1} -1 ) = 0.
\end{equation}

We claim that if $\pi_1 + \pi_0 = \pi_0^2 + \pi_0 - 1 < 0$, which implies $\pi_0 \leq \frac{1}{2}$, then all the terms in the summation (\ref{lemma3rec}) is negative, which is a contradiction. We have $\pi_M + \pi_{-M+1} -1 = \pi_0^{2^M} + \pi_0^{2^{-M+1}} -1$ by (\ref{eq:rec2}). Let $f(M) = 1 - \pi_0^{2^M} - \pi_0^{2^{-M+1}}$. Assume to the contrary that $\pi_0 \leq \frac{1}{2}$. Then $f(1) > 0$, and clearly we have $\lim_{M \rightarrow \infty} f(M) = 0$. We will show that $f(M) \geq f(M+1)$ for all $M \in \mathbb{Z}_{+}$. The derivative of $f(M)$ with respect to $M$ is
\begin{equation}\label{derivativeoffunc}
\dfrac{d f(M)}{d M} = -2^M \cdot \log(2) \cdot \pi_0^{2^M} \cdot \log(\pi_0) + 2^{-M+1} \cdot \log(2) \cdot \pi_0^{2^{-M+1}} \cdot \log(\pi_0),
\end{equation}
where $\log$ is the natural logarithm. Since $\pi_0 \leq \frac{1}{2}$ by assumption, we have $\log(2) \cdot \log(\pi_0) < 0$. Thus, (\ref{derivativeoffunc}) has the same sign with
\begin{equation}\label{derivativeoffunc2}
2^M \cdot \pi_0^{2^M}  - 2^{-M+1} \cdot \pi_0^{2^{-M+1}} = \dfrac{2^{2M-1} \cdot \pi_0^{2^M}  - \pi_0^{2^{-M+1}}}{2^{M-1}}.
\end{equation}
Since $2^M > 2M -1$ for all $M \geq 1$ and $2\pi_0 \leq 1$, we have

\begin{equation}\label{derivativeoffunc3}
2^{2M-1} \cdot \pi_0^{2^M}  - \pi_0^{2^{-M+1}} = (2 \pi_0)^{2M-1} \cdot \pi_0^{2^M - 2M +1} -  \pi_0^{2^{-M+1}} \leq \pi_0^{2^M - 2M +1} -  \pi_0^{2^{-M+1}} < 0
\end{equation}
for all $M > 1$, where the last inequality follows from the fact that $2^M - 2M +1 > 0$, $-M+1 < 0$, and $0 < \pi_0 < 1$. \

Now we prove that $\pi_0 < \frac{3}{4}$. Assume to the contrary that $\pi_0 \geq \frac{3}{4}$. Then clearly we have
\begin{equation}\label{lemma41}
\sum_{i \geq 1}^\infty \pi_i = \sum_{i \geq 1}^{\infty} \pi_0^{2^i} \geq 0.8.
\end{equation}
Now we want to find an upper bound for
\begin{equation}\label{lemma42}
\sum_{i \leq 0}^\infty (1-\pi_i) = \sum_{i \geq 0}^\infty (1-\pi_0^{2^{-i}}).
\end{equation}
Since $\frac{1-\pi_{0}^{2^{-i}}}{1-\pi_{0}^{2^{-i-1}}} = 1+ \pi_0^{2^{-i-1}}$ is an increasing function for $i \geq 0$, we can upper bound (\ref{lemma42}) by the following geometric series
\begin{equation}\label{lemma43}
(1-\pi_0) + (1- \pi_0^{2^{-1}}) + (1- \pi_0^{2^{-2}}) (1 + r + r^2 + ...) \leq 0.6,
\end{equation}
where $r = \frac{1-\pi_{0}^{2^{-2}}}{1-\pi_{0}^{2^{-3}}}$. But, (\ref{lemma41}) and  (\ref{lemma43}) contradict to (\ref{eq:det2}). \Halmos
\endproof

\begin{lemma}[Lipschitz condition]\label{appendixthmlipschitz}
The finite model satisfies the Lipschitz condition in $L_1$-distance.
\end{lemma}
\proof{Proof of Lemma \ref{appendixthmlipschitz}.}
Let $x = (x_i)_{i \in \mathbb{Z}}$ and $y = (y_i)_{i \in \mathbb{Z}}$ be two states of the finite model. By (\ref{eq:diffeq}), we have
\begin{align*}
  |F(x) - F(y)| &= \sum_{i=-\infty}^{\infty} | (x_{i-1}^d - x_i^d) - (x_i - x_{i+1}) - (y_{i-1}^d-y_i^d) + (y_i - y_{i+1})| \\
  &\leq  2\sum_{i=-\infty}^\infty |x_i^d - y_i^d| + 2\sum_{i=-\infty}^\infty |x_i - y_i|   \\
  &\leq (2+2d) \sum_{i=-\infty}^\infty |x_i - y_i|\\
  &=(2+2d)|x-y|,
\end{align*}
where in the first inequality we used the triangle inequality, and in the second inequality we used the expansion $(a-b)^n = (a-b)(a^{n-1} + a^{n-2}b + ... + ab^{n-2} + b^{n-1})$ and the fact that $0 \leq x_i,y_i \leq 1$ for all $i \in \mathbb{Z}$.
\Halmos
\endproof

\subsection{A simple asymmetric case}\label{appendix:asymmetric}

In this section, we further investigate whether the token system under the minimum token selection rule behaves well with asymmetric agents. In the context of kidney exchange, it is natural to ask whether large (or small) hospitals will have some advantage, or cause the system to be unstable. The system with asymmetric agents can also be  modeled as a density dependent Markov chain. However, note that even in the symmetric case, there was no closed form expression for the equilibrium point. Thus, instead of finding a closed form expression for the equilibrium point, the analogous differential equations for the asymmetric case can be used for numerical studies.

We discuss a very simple asymmetric setting with two types of agents referred to as $A$ and $B$, and let $d=2$ for simplicity. Assume that there are $n$ agents where $n$ is an even integer, and there are two types of agents: type $A$ and type $B$. Agents within the same type have the same service requesting and providing rate. Assume that there are $n/2$ many type $A$ agents and $n/2$ many type $B$ agents for simplicity. Define the service requesting distribution by $P=(p_i)_{i \in \mathcal{A}}$, where $p_i = p_A$ if agent $i$ is type $A$ and $p_i = p_B$ if agent $i$ is type $B$. This gives $\frac{n}{2}(p_A + p_B) = 1$. Similarly, define the service availability distribution by $Q=(q_i)_{i \in \mathcal{A}}$, where $q_i = q_A$ if agent $i$ is type $A$ and $q_i = q_B$ if agent $i$ is type $B$. This gives $\frac{n}{2}(q_A + q_B) = 1.$




Similar to the finite model, in order to fit the system with asymmetric agents to the definition of a density dependent Markov chain, we can assume that each agent has an exponential clock with mean $np_i$. Ticking of agent $i$'s clock corresponds to a service request by $i$. Note that the service requesting and service availability probabilities change as $n$ changes, but the ratio between the probabilities do not change. Thus, let us assume that $p_B = \alpha p_A$ and $q_B = \beta q_A$ for some constants $\alpha$ and $\beta$ which are independent of $n$.

Let $z_i^A(t)$ be the fraction of type $A$ agents with at least $i$ tokens at time $t$ among the type A agents and $z_i^B(t)$ be the fraction of type $B$ agents with at least $i$ tokens at time $t$ among the type B agents. Similar to the finite model, we will represent the state of the system by $\vec{z}(t) = (\vec{z^A}(t), \vec{z^B}(t))$ where $\vec{z^A}(t) = (..., z_{-1}^A(t), z_0^A(t), z_1^A(t), ...)$ and $\vec{z^B}(t) = (..., z_{-1}^B(t), z_0^B(t), z_1^B(t), ...)$. We drop the time index $t$ when the meaning is clear. Note that the initial state of the system is $\vec{z}(0) = (\vec{z^A}(0), \vec{z^B}(0))$, where $z_i^A(0) = z_i^B(0) = 1$ for all $i \leq 0$, and $z_i^A(0) = z_i^B(0) = 0$ for all $i \geq 1$.



We will denote the set of possible transitions from $\vec{k} = \frac{n \vec{z}}{2}$ by $L = \left\{ e_{ij}^A, e_{ij}^B, (e_i,-e_j), (-e_i, e_j) : i,j \in \mathbb{Z}, i \neq j \right\}$, where $e_{ij}^A = (e_{ij} , \vec{0})$ and $e_{ij}^B = (\vec{0}, e_{ij})$. Here, $e_{ij}$ is an infinite dimensional vector of all zeros except the $i$'th index (which corresponds to the index of $z_i^A(t)$ or $z_i^B(t)$) is $-1$ and the $j$'th index (which corresponds to the index of $z_j^A(t)$ or $z_j^B(t)$) is $1$, $e_i$ is an infinite dimensional vector of all zeros except the $i$'th index (which corresponds to the index of $z_i^A(t)$ or $z_i^B(t)$) is $-1$. For example, $e_{ij}^A$ corresponds to the transition where the service requester is a type $A$ agent with $i$ many tokens and the service provider is a type $A$ agent with $j-1$ many tokens. Using these notations, we can compute the following probabilities:

$\bullet$ The probability that the service requester is a type $A$ agent with $i$ many tokens is $\frac{1}{1+\alpha}(z_i^A - z_{i+1}^A)$. Denote this probability by $c_i^A$.

$\bullet$ The probability that the service requester is a type $B$ agent with $i$ many tokens is $\frac{\alpha}{1+\alpha}(z_i^B - z_{i+1}^B)$. Denote this probability by $c_i^B$.

$\bullet$ The probability that the service provider is a type $A$ agent with $j-1$ many tokens is $(\frac{1}{1+\beta})^2 ((z_{j-1}^A)^2 - (z_{j}^A)^2) + 2(\frac{1}{1+\beta})(\frac{\beta}{1+\beta}) (z_{j-1}^A- z_{j}^A)z_j^B + 2(\frac{1}{1+\beta})(\frac{\beta}{1+\beta}) (z_{j-1}^A- z_{j}^A) (z_{j-1}^B- z_{j}^B) \frac{1}{2}$. Denote this probability by $d_{j-1}^A$.

$\bullet$ The probability that the service provider is a type $B$ agent with $j-1$ many tokens is $(\frac{\beta}{1+\beta})^2 ((z_{j-1}^B)^2 - (z_{j}^B)^2) + 2(\frac{1}{1+\beta})(\frac{\beta}{1+\beta}) (z_{j-1}^B- z_{j}^B)z_j^A + 2(\frac{1}{1+\beta})(\frac{\beta}{1+\beta}) (z_{j-1}^A- z_{j}^A) (z_{j-1}^B- z_{j}^B) \frac{1}{2}$. Denote this probability by $d_{j-1}^B$.

Hence, we have $\beta_{e_{ij}^A} (\vec{z}) = c_i^A d_{j-1}^A, \beta_{e_{ij}^B} (\vec{z}) = c_i^B d_{j-1}^B,  \beta_{(e_i,-e_j)} (\vec{z}) = c_j^B d_{i-1}^A$ and $ \beta_{(-e_i,e_j)} (\vec{z}) = c_i^A d_{j-1}^B$. Clearly the condition (\ref{eq:jump}) is satisfied since the jump rate is bounded in the system, and given the constants $\alpha$ and $\beta$, the Lipschitz condition can be easily checked. Using Kurtz's theorem, the differential equations, which characterizes the infinite system with asymmetric agents can be found as

\begin{equation}\label{as1}
\dfrac{dz_i^A}{dt} = -c_i^A + d_{i-1}^A \;\; \text{for all} \;\;  i \in \mathbb{Z},
\end{equation}
\begin{equation}\label{as2}
\dfrac{dz_i^B}{dt} = -c_i^B + d_{i-1}^B \;\; \text{for all} \;\;  i \in \mathbb{Z}.
\end{equation}

Since agents start with 0 tokens and exchange one token at each transition of the system, the expected number of tokens agents have is $0$, and it can be translated as follows:
\begin{equation}\label{as3}
\sum_{i \geq 1}^\infty z_i^A-  \sum_{i \leq 0}^\infty (1 - z_i^A) + \sum_{i \geq 1}^\infty z_i^B-  \sum_{i \leq 0}^\infty (1 - z_i^B)  = 0.
\end{equation}



As we mentioned earlier, we are unable to find a closed form expression for the equilibrium point, but instead one may perform numerical studies using (\ref{as1}), (\ref{as2}) and (\ref{as3}).

We ran simulations in order to conduct comparative statics on having relatively more agents of one type and on the dominance of one type over the other. In both simulations, we fix the number of agents to $n=10$, and let the system run until $t=2 \cdot 10^7$. Let $f$ be the number of type $A$ agents and $p_{A,f,M}$ be the probability of in the long run, we have $-M \leq s_i^t \leq M$, where $M \in \mathbb{Z}^{+}$ and $i$ is a type $A$ agent. Similarly define $p_{B,f,M}$. In the first simulation, we fix $\alpha$ and $\beta$, and vary $f$ (see Tables \ref{sim1:t1} and \ref{sim1:t2}).
In the second simulation, we fix $f$ and vary $p_A$ (see Tables \ref{sim2:t1} and \ref{sim2:t2}).  Observe that in all simulations, having two types of agents does not create any instability as agents' number of tokens remain between $-4$ and $4$, with high probability.


  \begin{table}
        \begin{minipage}{0.5\textwidth}
            \centering
            \tabcolsep=0.11cm
            \begin{tabular}{|l|c|c|c|c|}
\hline
 $M / f$ & 2 & 4 & 6 & 8 \\
\hline
 1 & 0.6486 & 0.6476 & 0.6475 & 0.6469 \\ \cline{1-5}
 2 & 0.8787 & 0.8753 & 0.8734 & 0.8706 \\\cline{1-5}
 3 & 0.9523 & 0.9510 & 0.9515 & 0.9506 \\\cline{1-5}
 4 & 0.9769 & 0.9767 & 0.9777 & 0.9780 \\\cline{1-5}
 \hline
 \end{tabular}

\caption{The values for $p_{A,f,M}$ where\\ $p_B = 10p_A$ and  $q_B = 10q_A$.}
\label{sim1:t1}
        \end{minipage}
        \hfillx
        \begin{minipage}{0.5\textwidth}
            \centering
            \tabcolsep=0.11cm
           \begin{tabular}{|l|c|c|c|c|}
\hline
 $M / f$ & 2 & 4 & 6 & 8 \\
\hline
 1 & 0.6434 & 0.6410 & 0.6398 & 0.6405 \\\cline{1-5}
 2 & 0.8684 & 0.8645 & 0.8609 & 0.8587 \\\cline{1-5}
 3 & 0.9521 & 0.9512 & 0.9502 & 0.9495 \\\cline{1-5}
 4 & 0.9803 & 0.9809 & 0.9816 & 0.9824 \\\cline{1-5}
 \hline
 \end{tabular}

\caption{The values for $p_{B,f,M}$ where $p_B = 10p_A$ and  $q_B = 10q_A$.}
\label{sim1:t2}
        \end{minipage}
    \end{table}

 \begin{table}[h]
        \begin{minipage}{0.5\textwidth}
            \centering
           \tabcolsep=0.05cm
\begin{tabular}{|l|c|c|c|c|c|}
\hline
 $M / p_A$ & 0.02 & 0.04 & 0.06 & 0.08 & 0.1 \\
\hline
 1 & 0.6478 & 0.6469 & 0.6463 & 0.6453 & 0.6447\\\cline{1-6}
 2 & 0.8752 & 0.8742 & 0.8735 & 0.8726 & 0.8714\\\cline{1-6}
 3 & 0.9518 & 0.9520 & 0.9522 & 0.9527 & 0.9527\\\cline{1-6}
 4 & 0.9774 & 0.9779 & 0.9786 & 0.9794 & 0.9797\\\cline{1-6}
 \hline
 \end{tabular}

\caption{The values for $p_{A,5,M}$.}
\label{sim2:t1}
        \end{minipage}
        \hfillx
        \begin{minipage}{0.5\textwidth}
            \centering
\tabcolsep=0.05cm
          \begin{tabular}{|l|c|c|c|c|c|}
\hline
 $M / p_A$ & 0.02 & 0.04 & 0.06 & 0.08 & 0.1 \\
\hline
 1 & 0.6410 & 0.6418 & 0.6431 & 0.6437 & 0.6452\\\cline{1-6}
 2 & 0.8632 & 0.8657 & 0.8681 & 0.8697 & 0.8713\\\cline{1-6}
 3 & 0.9508 & 0.9516 & 0.9524 & 0.9525 & 0.9527\\\cline{1-6}
 4 & 0.9811 & 0.9810 & 0.9807 & 0.9801 & 0.9798\\\cline{1-6}
 \hline
 \end{tabular}
\caption{The values for $p_{B,5,M}$.}
\label{sim2:t2}
        \end{minipage}
    \end{table}

\end{document}